\def\confversion{0}
\def\ifconf{\ifnum\confversion=1}
\def\ifnotconf{\ifnum\confversion=0}

\newcommand{\confomit}[1]{
\ifnotconf
#1
\fi
}

\newcommand{\confequation}[1]{
\ifconf
{${#1}$}
\else
{\[{#1}\]}
\fi
}

\ifconf
\documentclass{sig-alternate}
\else
\documentclass[final, 11pt]{article}
\fi

\def\showauthornotes{1}
\def\showkeys{0}
\def\showdraftbox{0}

\usepackage{xspace,xcolor,enumerate}
\usepackage{amsmath,amsfonts,amssymb}
\usepackage{color,graphicx}

\ifnum\showkeys=1
\usepackage[color]{showkeys}
\fi

\usepackage{dsfont}

\ifnum\showauthornotes=1
\newcommand{\Authornote}[2]{{\sf\small\color{red}{[#1: #2]}}}
\newcommand{\Authorcomment}[2]{{\sf \small\color{gray}{[#1: #2]}}}
\newcommand{\Authorfnote}[2]{\footnote{\color{red}{#1: #2}}}
\else
\newcommand{\Authornote}[2]{}
\newcommand{\Authorcomment}[2]{}
\newcommand{\Authorfnote}[2]{}
\fi

\ifnum\showdraftbox=1
\newcommand{\draftbox}{\begin{center}
  \fbox{%
    \begin{minipage}{2in}%
      \begin{center}%
        \begin{Large}%
          \textsc{Working Draft}%
        \end{Large}\\
        Please do not distribute%
      \end{center}%
    \end{minipage}%
  }%
\end{center}
\vspace{0.2cm}}
\else
\newcommand{\draftbox}{}
\fi

\newtheorem{theorem}{Theorem}[section]

\newtheorem{lemma}[theorem]{Lemma}

\newtheorem{proposition}[theorem]{Proposition}
\newtheorem{corollary}[theorem]{Corollary}
\newtheorem{claim}[theorem]{Claim}

\def\FullBox{\hbox{\vrule width 6pt height 6pt depth 0pt}}

\def\qed{\ifmmode\qquad\FullBox\else{\unskip\nobreak\hfil
\penalty50\hskip1em\null\nobreak\hfil\FullBox
\parfillskip=0pt\finalhyphendemerits=0\endgraf}\fi}

\def\qedsketch{\ifmmode\Box\else{\unskip\nobreak\hfil
\penalty50\hskip1em\null\nobreak\hfil$\Box$
\parfillskip=0pt\finalhyphendemerits=0\endgraf}\fi}

\ifnum\confversion=0
\newenvironment{proof}{\begin{trivlist} \item {\bf Proof:~~}}
   {\qed\end{trivlist}}
\fi

\newenvironment{proofof}[1]{\begin{trivlist} \item {\bf Proof
#1:~~}}
  {\qed\end{trivlist}}

\def\to{\rightarrow}
\def\eps{\varepsilon}
\def\epsilon{\varepsilon}
\def\e{\epsilon}
\def\d{\delta}
\def\phi{\varphi}
\def\cal{\mathcal}

\def\ra{\rightarrow}

\renewcommand{\bar}{\overline} 
\newcommand{\ol}{\overline}

\newcommand{\mper}{\,.}
\newcommand{\mcom}{\,,}

\newcommand{\E}{{\mathbb E}}
\newcommand{\C}{{\mathbb C}}
\newcommand{\N}{{\mathbb{N}}}

\newcommand{\F}{{\mathbb F}}

\newcommand{\B}{\{0,1\}\xspace}
\newcommand{\pmone}{\{-1,1\}\xspace}

\newcommand{\indicator}[1]{\mathds{1}_{#1}}

\usepackage{nicefrac}

\newcommand{\abs}[1]{\ensuremath{\left\lvert #1 \right\rvert}}
\newcommand{\smallabs}[1]{\ensuremath{\lvert #1 \rvert}}
\newcommand{\norm}[1]{\ensuremath{\left\lVert #1 \right\rVert}}

\newcommand{\ip}[1]{\left\langle #1 \right\rangle}

\def\bfa{{\bf a}}

\newcommand{\Esymb}{\mathbb{E}}
\newcommand{\Psymb}{\mathbb{P}}

\DeclareMathOperator*{\ExpOp}{\Esymb}

\DeclareMathOperator*{\ProbOp}{\Psymb}
\renewcommand{\Pr}{\ProbOp}

\newcommand{\prob}[1]{\Pr\left[{#1}\right]}
\newcommand{\Prob}[2]{\Pr_{{#1}}\left[{#2}\right]}
\newcommand{\ex}[1]{\ExpOp\left[{#1}\right]}
\newcommand{\Ex}[2]{\ExpOp_{{#1}}\left[{#2}\right]}

\newfont{\inhead}{eufm10 scaled\magstep1}

\newcommand{\poly}{{\mathrm{poly}}}
\newcommand{\polylog}{{\mathrm{polylog}}}

\newcommand{\inparen}[1]{\left(#1\right)}             

\newcommand{\algofont}[1]{\fontshape{normal} \texttt{#1}}

\newcommand{\FindQuadratic}{\algofont{Find-Quadratic}\xspace}

\newcommand{\bogolyubov}{\algofont{Bogolyubov}\xspace}
\newcommand{\strongbogolyubov}{\algofont{Quasipolynomial-Bogolyubov}\xspace}

\newcommand{\Ztest}{\algofont{Z-Test}\xspace}
\newcommand{\Gtest}{\algofont{G-Test}\xspace}
\newcommand{\Xtest}{\algofont{X-Test}\xspace}

\newcommand{\uthreenorm}[1]{\norm{#1}_{U^3}}
\newcommand{\snorm}[1]{\norm{#1}_{S}}

\newcommand{\fhat}{\hat{f}}

\newcommand{\q}{\bar{q}}
\newcommand{\ortho}[1]{{#1}^{\perp}}
\def\ip#1{\langle #1\rangle}

\newcommand{\set}[1]{\left\{{#1}\right\}}
\newcommand{\condset}[2]{\set{{#1} \mid {#2} }}

\renewcommand{\a}{{{\mathbf{a}}}}
\renewcommand{\b}{{{\mathbf{b}}}}
\renewcommand{\c}{{{\mathbf{c}}}}

\newcommand{\ha}{{{\mathbf{\hat{a}}}}}
\newcommand{\hb}{{{\mathbf{\hat{b}}}}}
\newcommand{\rr}{\rho}
\newcommand{\rra}{\hat{\rho}_\a}
\newcommand{\rraa}[1]{\hat{\rr}_{#1}}
\newcommand{\G}{{\mathbf{G}}[\eps,2\delta]}
\newcommand{\Gb}{\G_\b}
\newcommand{\Gnew}{{\mathbf{G}}[\eps,\delta]}
\newcommand{\spec}{\mathrm{Spec}}
\newcommand{\Gparam}[1]{{\mathbf{G}}[#1]}
\newcommand{\codim}{\mathrm{codim}}

\newcommand{\oZ}{{{\indicator{Z_{\eta}}}}}
\newcommand{\oX}{{{\indicator{X}}}}

\title{%
Sampling-based proofs of almost-periodicity results \\and algorithmic applications
}%

\author{%
Eli Ben-Sasson
\thanks{Department of Computer Science, Technion, Haifa, Israel and MIT, Cambridge, MA.
 {\tt eli@cs.technion.ac.il}. The research leading to these results has received funding from the European Community's Seventh Framework
Programme (FP7/2007-2013) under grant agreement number 240258.}
\and
Noga Ron-Zewi
\thanks{Department of Computer Science, Technion, Haifa. {\tt nogaz@cs.technion.ac.il},  Research supported by a scholarship from the Israel Ministry of Science and Technology and by the US-Israel Binational Science
Foundation. Part of the research was conducted while the author was a visiting researcher at Microsoft Research New England, Cambridge, MA. 
The research leading to these results has also received funding from the European Community's Seventh Framework Programme (FP7/2007-2013) under grant agreement number 257575.
}
\and
Madhur Tulsiani\thanks{%
    Toyota Technological Institute at Chicago.
  }%
\and Julia Wolf\thanks{%
   Centre de Math\'ematiques Laurent Schwartz, \'Ecole Polytechnique, 91128 Palaiseau, France.
}%
}

\date{\today}

\begin{document}

\sloppy

\maketitle

\draftbox
\setcounter{page}{1}
\thispagestyle{empty}
\begin{abstract}
We give new combinatorial proofs of known almost-periodicity results for sumsets of sets with small doubling in the spirit of Croot and Sisask \cite{Croot:2010fk}, whose almost-periodicity lemma has had far-reaching implications in additive combinatorics. We provide an alternative (and $L^p$-norm free) point of view, which allows for proofs to easily be converted to probabilistic algorithms that decide membership in almost-periodic sumsets of dense subsets of $\F_2^n$.

As an application, we give a new algorithmic version of the \emph{quasipolynomial Bogolyubov-Ruzsa lemma} recently proved by Sanders \cite{Sanders10}. Together with the results by the last two authors \cite{TulsianiW11}, this implies an algorithmic version of the \emph{quadratic Goldreich-Levin theorem} in which the number of terms in the quadratic Fourier decomposition of a given function is quasipolynomial in the error parameter $\epsilon$, compared with an exponential dependence previously proved by the authors. It also improves the running time of the algorithm to have quasipolynomial dependence on $\epsilon$ instead of an exponential one.

We also show an application to the problem of finding large subspaces in sumsets of dense sets. Green showed in \cite{Green:2005zh} that the sumset of a dense subset of $\F_2^n$ contains a large subspace. Using Fourier analytic methods, Sanders \cite{Sanders:half} proved that such a subspace must have dimension $\Omega(\alpha n)$. We provide an alternative (and $L^p$ norm-free) proof of a comparable bound, which is analogous 
to a recent result of Croot, {\L}aba and Sisask \cite{Croot:2011we} in the integers.
\end{abstract}

\ifnum\confversion=0
\newpage
\tableofcontents
\thispagestyle{empty}
\setcounter{page}{1}
\fi

\section{Introduction}\label{sec:intro}

When Croot and Sisask introduced ``A probabilistic technique for finding almost-periods of convolutions" in 2009 \cite{Croot:2010fk}, it created quite a splash in the additive combinatorics community. Roughly speaking, their main result says that if $A \subseteq \F_2^n$ is a set whose sumset $A+A=\{a+a': a, a' \in A\}$ is small, then there exists a dense set $T$ such that the convolution $\indicator{A}*\indicator{A}( \cdot)$ of the indicator function of $A$ with itself and its translate $\indicator{A}*\indicator{A}(\cdot +t)$ are almost indistinguishable in the $L^2$ norm (or higher $L^p$ norms) for all $t \in T$. This set $T$ may then be referred to as the set of ``almost-periods".

Croot and Sisask's original proof used a simple sampling technique combined with tailbounds for a multinomial distribution, which Sanders replaced by the Marcinkiewicz-Zygmund inequality. Both made crucial use of $L^p$ norms, where in applications $p$ is taken to be very large (a function of the density $\alpha$ of a set $A \subseteq \F_2^n$ under investigation, such as $\log \alpha^{-1}$). 

Here we give a different proof of the Croot-Sisask lemma that proceeds entirely without recourse to $L^p$ norms, instead only relying on Chernoff-type tail estimates for sampling. It is our hope that this proof will appeal to a larger part of the theoretical computer science community than the currently existing ones, thereby increasing the likelihood of further novel applications of this lemma.

In the present paper we illustrate the use of this new technique by new and simplified proofs of several known results as well as an algorithmic application. Let us describe these in more detail.

\textbf{Applications.} In its original form, the Bogolyubov-Ruzsa lemma states that if $A \subseteq \F_2^n$ is a set of density $\alpha$, then $4A:=A+A+A+A$ contains a subspace of codimension at most $2\alpha^{-2}$. One of the first applications Croot and Sisask gave of their new technique was a \emph{weak} Bogolyubov-Ruzsa lemma, which asserted the existence of iterated sumsets of a dense set inside $4A$. It was quickly recognized by Sanders \cite{Sanders10} that the latter result could be boot-strapped, using a little Fourier analysis, to a quasipolynomial version of the Bogolyubov-Ruzsa lemma in which the codimension of the subspace that is found within $4A$ is polylogarithmic in the density of the set $A$ (so the size of this subspace is quasipolynomial in the density). This result has important implications for the bounds in Freiman's theorem which describes the structure of sets of integers with small sumsets \cite{ruzsa}, and the inverse theorem for the Gowers $U^3$ norm \cite{GrTu3}. It was also a crucial ingredient in Sanders's groudbreaking upper bound of $C (\log \log N)^5 N/\log N$ for the size of a subset of $\set{1,\ldots,N}$ not containing any 3-term arithmetic progressions \cite{Sanders:roth}.

In Section \ref{subsec:strong-bogol} we give a straightforward proof of Sanders's quasipolynomial Bogolyubov-Ruzsa lemma, specifically adapted to the setting of $\F_2^n$, which avoids the use of higher-order $L^p$ norms and instead relies exclusively on Chernoff-type tail bounds.

Next we present an algorithmic application. The original motivation for this paper lies with work by the last two authors on quadratic decomposition theorems. The aim of such theorems is to decompose any bounded function $f : \F_2^n\to \mathbb{C}$ as a sum $g+h$, where $g$ is quadratically uniform, in the sense that the Gowers $U^3$ norm $\|g\|_{U^3}$ is small, and $h$ is quadratically structured, in the sense that it is a bounded sum of quadratically structured objects. These types of decompositions constitute a higher-order analogue of classical Fourier decompositions, and had previously been obtained in an abstract and non-constructive way (either using a form of the Hahn-Banach theorem \cite{GW2}, or a so-called energy increment approach \cite{GrML}).

In \cite{TulsianiW11}, the authors gave a probabilistic algorithm that, given any function $f: \F_2^n \ra \C$, would with high probability compute, in time polynomial in $n$, a quadratic decomposition for that function with a specified $U^3$ error $\epsilon$. This essentially amounts to computing a ``quadratic Fourier decomposition" for $f$, and was therefore termed a \emph{quadratic Goldreich-Levin theorem} in analogy with the well-known linear case \cite{GoldreichL89}. The quadratic Goldreich-Levin algorithm consisted of two parts: a deterministic part which is able to construct the quadratically structured part of $f$ under the assumption that we have an algorithm which provides some quadratic phase functions that $f$ correlates with (if there is no such phase function, we just set $g=f$). The algorithm for finding a quadratic phase function, which constitutes the second part of the overall algorithm, is basically an algorithmic version of the proof of the inverse theorem for the $U^3$ norm, which states that if a bounded function $f$ has large $U^3$ norm, then it correlates with a quadratic phase.

As stated above, the Bogolyubov-Ruzsa lemma is crucial in the proof of the inverse theorem, and it
should not come as surprise that a new proof with a quantitative improvement has implications for
the efficiency of the quadratic Goldreich-Levin algorithm outlined above. In Section
\ref{subsec:bralgo} we tie the techniques developed in the earlier sections of the paper into the
algorithm given in \cite{TulsianiW11} to obtain an improvement in the running time (from exponential
to quasipolynomial in the quadratic uniformity parameter $\epsilon$) as well as in the number of
terms that are obtained in the final quadratic decomposition (with a similar improvement in the
dependence on $\epsilon$). One of the main difficulties, encountered already in \cite{TulsianiW11},
is that the individual subroutines in the quadratic Goldreich-Levin algorithm, which correspond to
algorithmic versions of theorems in additive combinatorics, are probabilistic in nature. Since they
are applied in sequence, this means that the input for the next subroutine comes with a certain
amount of noise, and it is therefore necessary to prove robust algorithmic versions of the theorems
from additive combinatorics. This applies in particular to the quasipolynomial Bogolyubov-Ruzsa
lemma, for which we give robust version in this paper. The detailed introduction of the key concepts in additive combinatorics and quadratic Fourier analysis is postponed to the start of Section \ref{sec:app1}.

Our final application concerns the problem of finding large subspaces within sumsets of a dense set. Green \cite{Green:2005zh} had shown that if $A \subseteq \F_2^n$ has density $\alpha$, then $A+A$ contains a subspace of dimension $\Omega(\alpha^2n)$. Sanders proved in \cite{Sanders:half} using a Fourier-iteration lemma that this subspace must be of dimension at least $\Omega(\alpha n)$, and remarked that a bound of comparable strength follows implicitly from the techniques of Croot, {\L}aba and Sisask \cite{Croot:2011we}, who addressed the more general problem in the integers, asking for long arithmetic progressions in sumsets of dense sets.

In Section \ref{sec:app2} we provide a simplified proof of the Croot-{\L}aba-Sisask bound, again avoiding Fourier analysis and using instead our sampling approach to Croot-Sisask almost-periodicity in $\F_2^n$. It requires a more careful analysis of our sampling technique, which we shall give in detail in the appendix. We do not address the question in the integers, nor the non-abelian case, of which this is a toy version. 

\textbf{Chernoff vs. $L^p$ norms.}
It is of course well known that $L^p$ bounds and Chernoff's inequality are, in a certain sense, equivalent. Specifically,
a random variable $X$ obeys a Chernoff-type tail bound of the form
\[ \Pr[|X| \geq t \|X\|_2 ] \leq C \exp(-\Omega(t^2))\]
if and only if its $L^p$ norm satisfies
\[ \| X\|_p \leq C \sqrt{p} \|X\|_2\]
for all $p \in [2,\infty)$, the latter representing a Khinchine-type inequality (from which Marcinkiewicz-Zygmund can be derived). For a proof of this statement we refer the reader to the excellent lecture notes by Sanders \cite{Sanders:to}.

We therefore do not claim that our proof of Croot and Sisask's almost-periodicity results is radically new. However, we do think it writes itself rather naturally in the special case of $\F_2^n$, and it lends itself more readily to applications in that setting. Moreover, these results can also be ``algorithmified'' in what is in our opinion a more natural way. 

\textbf{Acknowledgements.} The last two authors would like to thank Tom Sanders for numerous helpful remarks and discussions. The first two authors would like to thank Shachar Lovett for sharing with them his view of Sanders's proof of the quasipolynomial Bogolyubov-Ruzsa lemma (cf. \cite{Lovett12}).

\section{Preliminaries}\label{sec:prelim}

In this section we fix our notation and collect some results that we shall use throughout the paper. Fundamental to our approach will be the following Chernoff-type tail bound for sampling \cite{TaoVu}.

\begin{lemma}[Hoeffding bound for sampling]\label{lem:hoeffding-sample}
If $\bf X$ is a random variable with $\abs{{\bf X}} \leq 1$ and $\hat{\mu}$ is the empirical
average obtained from $t$ samples, then
\[ \prob{\abs{\ex{{\bf X}} - \hat{\mu}} ~>~ \gamma} ~\leq~ 2\exp(- 2\gamma^2 t) . \]
\end{lemma}

Throughout the paper we shall make use of the discrete Fourier transform, which we define as follows. For $f: \F_2^n \to \C$, let
\[ \widehat{f}(t)= \E_{x \in \F_2^n} f(x)(-1)^{x\cdot t}\]
for any $t \in \widehat{\F_2^n}=\F_2^n$, where $\E_{x \in \F_2^n}$ simply stands for the normalized sum $2^{-n} \sum_{x \in \F_2^n}$
and $x \cdot t = \sum_{i=1}^n x_i t_i$ for a pair of vectors $x = (x_1,\ldots,x_n)$, $t = (t_1,\ldots,t_n)$. The inversion formula states that
\[ f(x)=\sum_{t \in \F_2^n} \widehat{f}(t) (-1)^{x \cdot t}\]
for all $x \in \F_2^n$, and Parseval's identity takes the form
\[ \ip{f,g} =\ip{\widehat{f},\widehat{g}},\]
where the inner product is defined as $\ip{f,g}=\E_{x \in \F_2^n} f(x)\ol{g(x)}$ in physical space, and $\ip{\widehat{f},\widehat{g}}=\sum_{t \in \F_2^n} \widehat{f}(t) \ol{\widehat{g}(t)}$ on frequency space, for any two functions $f,g :\F_2^n \to \C$.

Finally, the convolution of two such functions is defined by
\[f*g(x)=\E_{y \in \F_2^n} f(y)g(x-y),\]
and the fact that the Fourier transform diagonlizes the convolution operator is expressed via the idenity
\[ \widehat{f*g }(t)=\widehat{f}(t) \widehat{g}(t),\]
which holds for all $t \in \F_2^n$.

The set of large Fourier coefficients determines the value of a function to a significant extent, and for many arguments it is important to be able to estimate its size and determine its structure. For a function $f:\F_2^n \to \mathbb{C}$, let
\begin{equation}\label{eq:spec-def}
 \spec_\rho(f)= \{ t \in \F_2^n: |\widehat{f}(t)| \geq \rho \|f\|_1\}.\end{equation}

For a subset $A \subseteq \F_2^n$ we let $\indicator{A}$ denote the indicator function of $A$ and $\mu_A$ denote the function $\indicator{A} \cdot (2^n/|A|)$ so that
$\mathbb{E}_{x \in \F_2^n} [\mu_{A}(x)]=1$.
In the special case where $f = \indicator{A}$ for a subset $A \subseteq \F_2^n$ of density $\alpha$,
Parseval's identity tells us that $|\spec_\rho(\indicator{A})| \leq \rho^{-2} \cdot  \alpha^{-1}$. A more precise result is known: Chang's theorem \cite{Chang:2002vn} states that $\spec_\rho(\indicator{A})$ is in fact contained in a subspace of dimension at most $C\rho^{-2} \log \alpha^{-1}$.

\begin{theorem}[Chang's theorem]\label{thm:chang}
Let $\rho \in (0,1]$ and $A \subseteq \F_2^n$. Then there is a subspace $V$ of $\F_2^n$ such that $\spec_\rho(\indicator{A}) \subseteq V$ and
\[ \mathrm{dim}\big(V \big)\leq 8 \frac{ \log(2^n/|A|)}  {\rho^2}.\]
\end{theorem}

For an elegant recent proof of this result using entropy, see Impagliazzo et al. \cite{Impagliazzo:2012kx}.

Finally, for two real numbers $\alpha,\beta$ we write $\alpha \approx_\eps \beta$ to denote $|\alpha-\beta|\leq \eps$ and if $|\alpha-\beta|>\eps$ we write $\alpha\not \approx_\eps \beta$. All logarithms in this paper are taken to base 2.

\section{Sampling-based proofs of almost-periodicity results}\label{sec:ap}

For comparison, we give the precise statement of the original result of Croot and Sisask (Proposition 1.3 in \cite{Croot:2010fk}). Since it is valid for general groups $G$, it is written in multiplicative notation.

\begin{proposition}[Croot-Sisask Lemma\confomit{, $L^p$ local version}]\label{prop:cs}
Let $\eps>0$ and let $m\geq 1$ be an integer. Let $G$ be a group and let $A,B \subseteq G$ be finite
subsets such that $|B \cdot A| \leq K |B|$. Then there is a
set $X \subseteq A$ of size $|X| \geq |A|/(2K)^{50 m /\eps}$ such that for each $x \in X X^{-1}$,
\ifconf
\begin{align*} 
 &\|\indicator{A}*\indicator{B}(yx)-\indicator{A}*\indicator{B}(y)\|_{2m}^{2m} \leq \\
&\max\{\eps^m|AB||B|^m, \| \indicator{A}*\indicator{B}\|_m^m\}\eps^m|B|^m .
\end{align*}
\else
\[  \|\indicator{A}*\indicator{B}(yx)-\indicator{A}*\indicator{B}(y)\|_{2m}^{2m} \leq
\max\{\eps^m|AB||B|^m, \| \indicator{A}*\indicator{B}\|_m^m\}\eps^m|B|^m .\]
\fi
\end{proposition}

In the next section we prove our version of the Croot-Sisask lemma, given as Proposition \ref{prop:ap-sumsets} below. In Section \ref{subsec:ap-subspace} we give a modified version in which the resulting set of almost-periods, appearing as $XX^{-1}$ in Proposition \ref{prop:cs} above, can in fact taken to be a subspace, which is what we need in applications.

\subsection{Croot-Sisask almost-periodicity}\label{subsec:cs}

To state our result concisely, we define the following measure which will play a central role in what follows. Given subsets $A,B\subset \F_2^n$ where $A$ is finite, define the {\em measure of additive containment} $\rho_{A\to B}:\F_2^n \to[0,1]$ by
\ifconf
\begin{align}\label{eq:mu-definition}
 \nonumber \rho_{A\to B}(y):= \Pr_{a\in A}\left[y+a\in B\right] &= \frac{|(y+A)\cap B|}{|A|} \\
&= \mu_A * \indicator{B}(y),
\end{align}
\else
\begin{equation}
  \label{eq:mu-definition}
  \rho_{A\to B}(y):= \Pr_{a\in A}\left[y+a\in B\right] = \frac{|(y+A)\cap B|}{|A|} = \mu_A * \indicator{B}(y),
\end{equation}
\fi
for each $y\in \F_2^n$.
Notice that $\rho_{A\to B}(y)=1$ when $y+A\subseteq B$ and $\rho_{A\to B}(y)=0$ when $(y+A)\cap B=\emptyset$.

\begin{proposition}[Almost-periodicity\confomit{ of sumsets}]\label{prop:ap-sumsets}
  If $A\subset\F_2^n$ satisfies $|2A|\leq K|A|$, then for every integer $t$ and set $B\subseteq \F_2^n$ there exists a set $X$ with the following properties.
  \begin{enumerate}
  \item\label{p1} The set $X$ is contained in an affine shift of $A$.
  \item\label{p2} The size of $X$ is at least $|A|/(2K^{t-1})$.

  \item\label{p3} For all $x\in X$ and for all subsets $S \subseteq \F_2^n$,
\ifconf
  \begin{align}
    \label{eq11}
 \nonumber   &\Pr_{y\in S}\left[\rho_{A\to B}(y)\approx_{2\eps} \rho_{A\to B}(y+x)\right] \\
&~~\geq 1-8 \frac{|A+B|}{|S|}\cdot \exp\left(-2\eps^2 t\right).
  \end{align}
\else
  \begin{equation}
    \label{eq11}
    \Pr_{y\in S}\left[\rho_{A\to B}(y)\approx_{2\eps} \rho_{A\to B}(y+x)\right]\geq 1-8 \frac{|A+B|}{|S|}\cdot \exp\left(-2\eps^2 t\right).
  \end{equation}
\fi
  \end{enumerate}
\end{proposition}

Our proof differs from the original proof of Croot and Sisask in that it disposes of $L^p$-norms and tail bounds for a multinomial distribution (or the Marcinkiewicz-Zygmund inequality), and replaces them with sampling arguments relying on the Chernoff-Hoeffding
bound.

We sketch the proof before giving the technical details.
To obtain $X$ we replace $\rho_{A\to B}$ by an estimator function computed by taking a sequence of $t$ independent random samples distributed uniformly over $A$. Denoting the sample sequence by $\a=(a_1,\ldots, a_t)$, we estimate $\rho_{A\to B}(y)$ by the fraction of $a_i\in \a$ satisfying $y+a_i\in B$. Denote the estimator function corresponding to $\a$ by $\rra$. Fixing $y$, the Chernoff-Hoeffding bound says that the probability that $\rra(y)$ differs from $\rho_{A\to B}(y)$ by more than $\eps$, i.e., the probability of the event ``$\rho_{A\to B}(y)\not\approx_\eps \rra(y)$'' when $\a=(a_1,\ldots,a_t)$ is distributed uniformly over $A^t$, is at most $\exp\left(-\Omega\left(\eps^2 t\right)\right)$.

The key observation in the construction of the set $X$ is that there are many pairs of good estimator-sequences $\a=(a_1,\ldots,a_t), \hat \a=(\hat a_1,\ldots,\hat a_t)$ for which there exists a ``special'' element $x\in \F_2^n$ such that  $\hat \a=x+\a$, where $x+\a := (x+a_1,\ldots,x+a_t)$. Such $x$ can be justly called ``special'' for the following reason. Call $y$ ``good'' if both of the following conditions hold,
\begin{equation}
  \label{y-good}
  \rho_{A\to B}(y)\approx_\eps \hat \rho_{\hat \a}(y) \quad \mbox{ and } \quad\rho_{A\to B}(y+x)\approx_\eps \rraa{\a}(y+x).
\end{equation}
Now if $\hat \a=x+\a$, then we have \[\hat \rho_{\hat \a}(y)=\hat \rho_{\hat \a}(y+x+x)=\hat \rho_{\hat \a+x}(y+x)=\rraa{\a}(y+x),\]
and combining this with (\ref{y-good}) implies that for ``good'' $y$ we have $\rho_{A\to B}(y)\approx_{2\eps} \rho_{A\to B}(y+x)$. Thus, to prove the proposition we only need to bound from below the number of ``special'' elements $x$, which is done based on the assumption that $A$ has small doubling. We now give the formal proof.

\begin{proofof}{of Proposition~\ref{prop:ap-sumsets}}
To simplify notation let $\rr(y):=\rho_{A\to B}(y)$,
and for a sequence $\a=(a_1,\ldots,a_t)$ of length $t$, define the {\em $\a$-estimator} of $\rr$ to be the function $\rra:\F_2^n\to [0,1]$ defined for $y\in \F_2^n$ by
\[\rra(y):=\frac{\left|\condset{y+a_i\in B}{i=1,\ldots,t}\right|}{t}.\]

We say that $\a$ is an {\em $\eps$-good estimator} for $y$ if $\rr(y)\approx_\eps \rra(y)$.

Fix $y \in \F_2^n$. Our first step towards constructing $X$ is to show that most sample-sequences from $A$ are $\epsilon$-good for $y$, provided that $t$, the sample size, is large enough with respect to $1/\eps$. Let $Y_i$ be the indicator random variable for the event ``$y+a_i\in B$'' when $a_i$ is chosen uniformly at random from $A$. Then $\rra(y)=\frac{1}{t}\sum_{i=1}^t Y_i$ is the average of $t$ i.i.d. indicator random variables each having mean $\rr(y)$, so the Chernoff-Hoeffding bound (Lemma~\ref{lem:hoeffding-sample})
implies that for each $y \in \F_2^n$,
\begin{equation}
  \label{eq:t}
  \Pr_{\a\in A^t}\left[\rr(y)\not \approx_\eps \rra(y)\right]\leq 2\exp\left(-2\eps^2 t\right).
\end{equation}
Now we proceed to show that most $\a \in A^t$ are $\eps$-good estimators for most $y$. Let $Z_{\a}$ be the random variable measuring the fraction of $y\in A+B$ for which $\a$ is an $\eps$-good estimator, that is,
\[Z_\a:=\Pr_{y\in A+B}\left[\rr(y)\approx_\eps \rra(y)\right].\]
Setting $\delta=2\exp\left(-2\eps^2 t\right)$, we conclude from (\ref{eq:t}) via linearity of expectation that
\[\E_{\a\in A^t}\left[Z_{\a}\right]\geq 1-\delta.\]
Markov's inequality now shows that at least half of the sequences $\a \in A^t$ are $\eps$-good estimators for all but a ($2\delta$)-fraction of $y\in A+B$, in which case we say that $\a$ is an {\em $(\eps,2\delta)$-good estimator} for $\rr$. Denote by $\G\subset A^t$ the set of these sequences,
\[\G=\condset{\a\in A^t}{\Pr_{y\in A+B}\left[\rr(y)\approx_\eps \rra(y)\right]\geq 1-2\delta}.\]

To obtain $X$ we partition $\G$ as follows. Define a mapping $\phi:A^t\mapsto \set{0}\times (2A)^{t-1}$ by shifting a sequence $\a=(a_1,\ldots,a_t)$ by its first element $a_1$,
\begin{equation}\label{eq:phi}\phi(\a)=\a+a_1 := \left(a_1+a_1, a_1+a_2,\ldots, a_1+a_t\right)\end{equation}

Then $\phi$ maps the set $\G$, which has size at least $|A|^t/2$, into a set of size
$|2A|^{t-1}\leq (K|A|)^{t-1}$ so by the pigeonhole principle, there is a subset $\Gb \subset \G$ that is mapped to the same element $\b=(0,b_2,\ldots,b_t)$. In addition, this subset is pretty large,
\begin{equation}\label{eq:size}
|\Gb|\geq \frac{|A|^t}{2 K^{t-1} |A|^{t-1}} = \frac{|A|}{2K^{t-1}}.
\end{equation}

Finally, fix an arbitrary $\ha=(\hat{a}_1,\ldots,\hat{a}_t)\in \Gb$ and set
\[X=\condset{\hat{a}_1+a_1}{(a_1,\ldots,a_t)\in \Gb}.\]

To complete our proof we show that $X$ has the three properties listed in the statement of the lemma.
\begin{enumerate}
  \item By definition, $X\subseteq \hat{a}_1+A$.
  \item The mapping $\Gb\mapsto X$ given by $(a_1,\ldots,a_t)\mapsto \hat{a}_1+a_1$ is invertible, because both $\hat{a}_1$ and $\b$ are fixed. Hence
      $|X|=|\Gb|$ and the size of $X$ is bounded from below using (\ref{eq:size}).
  \item Suppose $x=\hat{a}_1+a_1$, where $a_1$ is the first element of an $(\eps,2\delta)$-good estimator $\a=(a_1,\ldots,a_t)\in \Gb$.
      The key observation is that $\ha+x=\a$. Indeed, the definition of $\Gb$ implies $\phi(\ha)=\phi(\a)$, so using (\ref{eq:phi}) we have
  \[\hat{a}_1+\hat{a}_i\ =\ a_1+a_i, \quad i=1,\ldots,t,\]
  which, rearranging, comes out to
  \[a_i = x+\hat{a}_i, \quad i=1,\ldots,t.\]
 In other words, $\ha+x=\a$ as claimed.

  Recalling that $\ha$ is an $(\eps,2\delta)$-good estimator,
  we know that for all but a $2\delta$-fraction of $y\in A+B$,
  \begin{equation}\label{eq:muone}
  \rr(y)\approx_\eps \rraa{\ha}(y)
  \end{equation}
  and (\ref{eq:muone}) also holds for all  $y \notin A+B$ since in this case   $\rraa{\ha}(y) =\rr(y) =0$. Hence we have that (\ref{eq:muone}) holds for all
  but a $\left(2\delta |A+B|/|S|\right)$-fraction of $y\in S$.

 Similarly, since $\a$ is an $(\eps,2\delta)$-good estimator, we have that
 \begin{equation}\label{eq:mutwo}
  \rr(y+x)\approx_\eps \hat \rho_{\a}(y+x)
  \end{equation}
   for all but a $\left(2\delta |A+B|/|x+S|\right)$-fraction of $y\in S$.
  Using a union bound and the fact that $|S+x| = |S|$, we find that for all but a $\left(4\delta |A+B|/|S|\right)$-fraction of $y\in S$ both (\ref{eq:muone}) and (\ref{eq:mutwo}) hold. For such $y$ we conclude that $\rr(y)\approx_{2\eps} \rr(y+x)$ using the triangle inequality and the fact that $\rraa{\ha}(y)=\rraa{\ha+x}(y+x)=\rraa{\a}(y+x)$.

\end{enumerate}
This completes the proof of the proposition.
\end{proofof}

By an inductive application of Proposition \ref{prop:ap-sumsets} one can prove its following iterated version.

\begin{corollary}[\ifconf Almost periodicity, iterated\else Almost-periodicity of sumsets, iterated\fi]\label{lem:ap-sumsets-iterated}
  If $A\subset\F_2^n$ satisfies $|2A|\leq K|A|$, then for every integer $t$ and set $B\subseteq \F_2^n$ there exists a set $X$ with the following properties.
  \begin{enumerate}
  \item\label{p1} The set $X$ is contained in an affine shift of $A$.
  \item\label{p2} The size of $X$ is at least $|A|/(2K^{t-1})$.

  \item\label{p3} For all $x_1, \ldots, x_{\ell} \in X$ and for all subsets $S \subseteq \F_2^n$,
\ifconf
  \begin{align}
    \label{eq11}
\nonumber    &\Pr_{y\in S}\left[\rho_{A\to B}(y)\approx_{2\eps \ell} \rho_{A\to B}(y+x_1+\ldots +x_{\ell})\right]\\&~~\geq 1-8 \ell \frac{|A+B|}{|S|}\cdot \exp\left(-2\eps^2 t\right).
  \end{align}
\else
  \begin{equation}
    \label{eq11}
    \Pr_{y\in S}\left[\rho_{A\to B}(y)\approx_{2\eps \ell} \rho_{A\to B}(y+x_1+\ldots +x_{\ell})\right]\geq 1-8 \ell \frac{|A+B|}{|S|}\cdot \exp\left(-2\eps^2 t\right).
  \end{equation}
\fi
  \end{enumerate}
\end{corollary}

\begin{proof}
The proof is by induction on $\ell$. Proposition \ref{prop:ap-sumsets} establishes the case $\ell =1$. For the induction step, suppose that the lemma holds for some integer $\ell \geq 1$ with a set $X \subseteq \F_2^n$. We shall show that the same set $X$ satisfies the above requirements for $\ell +1$.

Let $\delta : = 8 (|A+B|/|S|)\cdot \exp\left(-2\eps^2 t\right)$.
By the induction hypothesis, for at least a ($1- \ell \delta$)-fraction of $y \in S$, it is true that
$\rho_{A\to B}(y)\approx_{2\eps \ell} \rho_{A\to B}(y+x_1+\ldots +x_{\ell})$.
 The case $\ell =1$ implies that for at least a ($1- \delta$)-fraction of $y \in S$, it is true that
 $\rho_{A\to B}(y+x_1+\ldots+x_{\ell})\approx_{2\eps} \rho_{A\to B}(y+x_1+\ldots +x_{\ell+1})$.
 Thus by a union bound we have that for at least a
($1-(\ell+1)\delta$)-fraction of $y \in S$ we have both $\rho_{A\to B}(y)\approx_{2\eps \ell } \rho_{A\to B}(y+x_1+\ldots +x_{\ell})$ and $\rho_{A\to B}(y+x_1+\ldots+x_{\ell})\approx_{2\eps } \rho_{A\to B}(y+x_1+\ldots +x_{\ell+1})$. The proof is completed by noting that by the triangle inequality, for each such $y$, we also have $\rho_{A\to B}(y)\approx_{2\eps (\ell+1)} \rho_{A\to B}(y+x_1+\ldots +x_{\ell+1})$.
\end{proof}

\subsection{Almost-periodicity over a subspace}\label{subsec:ap-subspace}

For applications one would like a version of Proposition \ref{prop:ap-sumsets} in which the set $X$ of periods is in fact a subspace. It was observed by Sanders \cite{Sanders10} that one can use iterated almost-periodicity statements such as Corollary \ref{lem:ap-sumsets-iterated}, combined with some Fourier analysis, to obtain such a subspace. In this section, we use Sanders's argument to deduce the following statement from Corollary \ref{lem:ap-sumsets-iterated}.

\ifconf 
\begin{corollary}\label{lem:ap-sumsets-subspace}
\else
\begin{corollary}[Almost-periodicity of sumsets over a subspace]\label{lem:ap-sumsets-subspace}
\fi
If $A\subset\F_2^n$ is a subset of density $\alpha$, then for every integer $t$ and set $B\subseteq \F_2^n$ there exists a subspace $V$ of codimension $\codim(V) \leq 32 \log (2/\alpha^t)$ with the following property.

For every $v \in V$, for all subsets $S \subseteq \F_2^n$ and for every $\epsilon, \eta >0$ and
integer $\ell$,
\ifconf
  \begin{align}
    \label{eq:subspace-main}
\nonumber    &\Pr_{y\in S}\left[\rho_{A\to B}(y)\approx_{\epsilon'}\rho_{A\to B}(y+v)\right] \\
&~~\geq 1-16 \frac{\ell} {\eta} \frac{|A+B|}{|S|}\cdot \exp\left(-2\eps^2 t\right),
  \end{align}
\else
  \begin{equation}
    \label{eq:subspace-main}
    \Pr_{y\in S}\left[\rho_{A\to B}(y)\approx_{\epsilon'}\rho_{A\to B}(y+v)\right]\geq 1-16 \frac{\ell} {\eta} \frac{|A+B|}{|S|}\cdot \exp\left(-2\eps^2 t\right),
  \end{equation}
\fi
  where $\epsilon' = 4\eps \ell + 2 \eta +2^{-\ell} \sqrt{|B|/|A|}$.
  \end{corollary}

As we shall see in Sections \ref{sec:app1} and \ref{sec:app2}, the proof of the quasipolynomial Bogolyubov-Ruzsa lemma (Theorem \ref{thm:bogolyubov-algo}) follows easily from the above lemma, and Green's theorem on the existence of subspaces in sumsets of dense sets (Theorem \ref{thm:sumsets}) follows easily from a refinement of the above corollary (which we will give as Corollary \ref{lem:ap-sumsets-subspace-refined} below).
Note that for the proof of Corollary \ref{lem:ap-sumsets-subspace} we need the stronger assumption that $A$ has density at least $\alpha$ in $\F_2^n$, instead of the doubling hypothesis $|2A| \leq K|A|$.

The idea of the proof of Corollary \ref{lem:ap-sumsets-subspace} is the following. Let $X$ be the
subset guaranteed by Corollary \ref{lem:ap-sumsets-iterated} for $K= 1/\alpha$, and define the
subspace $V$ as $V = \spec_{1/2} (X) ^{\perp}$ (see Section \ref{sec:prelim} for the definition of
$\spec_\rho$). The intuition is that if $X$ were a subspace then $\spec_{1/2}(X) = V^{\perp}$, and hence
$V = X$. Thus $V$ serves as an ``approximate subspace'' for $X$. Since $A$ is dense in $\F_2^n$, by
Corollary \ref{lem:ap-sumsets-iterated} we also have that $X$ is dense in $\F_2^n$ and hence
Chang's theorem (Theorem \ref{thm:chang}) implies that the subspace $V$ is also dense in $\F_2^n$
(this is the only place where we need the stronger assumption on the density of $A$).

In order to show that (\ref{eq:subspace-main}) holds we first show, using Corollary \ref{lem:ap-sumsets-iterated}, a simple averaging argument and the triangle inequality, that for most $y \in S$,
\begin{equation}\label{eq:subspace1}
\Ex{x_1,\ldots,x_{\ell} \in X}{\rho_{A\to B}(y+x_1+\ldots+x_{\ell})} \approx_{2 \epsilon\ell + \eta} \rho_{A\to B}(y)\mper
\end{equation}
Similarly, for all $v \in V$ and for most $y \in S$,
\ifconf
\begin{align}\label{eq:subspace2}
\nonumber &\Ex{x_1,\ldots,x_{\ell} \in X}{\rho_{A\to B}(y+v+x_1+\ldots+x_{\ell})} \\
&~~\approx_{2 \epsilon \ell + \eta}~~ \rho_{A\to B}(y+v)\mper
\end{align}
\else
\begin{equation}\label{eq:subspace2}
\Ex{x_1,\ldots,x_{\ell} \in X}{\rho_{A\to B}(y+v+x_1+\ldots+x_{\ell})} \approx_{2 \epsilon \ell + \eta} \rho_{A\to B}(y+v)\mper
\end{equation}
\fi

We then use Fourier analysis, following Sanders's argument closely, to show that for {\em all} $y
\in \F_2^n$,
\ifconf
\begin{align}\label{eq:subspace3}
\nonumber &\Ex{x_1,\ldots,x_{\ell} \in X}{\rho_{A\to B}(y+x_1+\ldots+x_{\ell})} \\
&~~\approx_{2^{-\ell}
  \sqrt{|B|/|A|}} \Ex{x_1,\ldots,x_{\ell} \in X}{\rho_{A\to B}(y+v+x_1+\ldots+x_{\ell})} \mcom
\end{align}
\else
\begin{equation}\label{eq:subspace3}
\Ex{x_1,\ldots,x_{\ell} \in X}{\rho_{A\to B}(y+x_1+\ldots+x_{\ell})} \approx_{2^{-\ell}
  \sqrt{|B|/|A|}} \Ex{x_1,\ldots,x_{\ell} \in X}{\rho_{A\to B}(y+v+x_1+\ldots+x_{\ell})} \mcom
\end{equation}
\fi
where $v$ is again an arbitrary element of $V$. The final conclusion follows from (\ref{eq:subspace1}), (\ref{eq:subspace2}) and (\ref{eq:subspace3}) using the union bound and the triangle inequality.
We start by establishing (\ref{eq:subspace1}) and (\ref{eq:subspace2}).

\begin{lemma}\label{lem:subspace-averaging}
Let $\eps, \delta >0$, and let $A,B,X,S \subseteq \F_2^n $ be such that for all $x_1,\ldots,x_{\ell} \in X$,
\[
\Prob{y\in S}{\rho_{A\to B}(y)\approx_{\epsilon}\rho_{A\to B}(y+x_1+\ldots+x_{\ell})} 
\geq 1-\delta \mper
\]
Then for every $\eta >0$ we have that
\ifconf
\begin{align*}
&\Prob{y\in S}{\rho_{A\to B}(y)\approx_{\epsilon+ \eta} \Ex{x_1,\ldots,x_{\ell} \in
    X}{\rho_{A\to B}(y+x_1+\ldots+x_{\ell})} } \\
&~~\geq 1 - \delta/\eta \mper 
\end{align*}
\else
\[
\Prob{y\in S}{\rho_{A\to B}(y)\approx_{\epsilon+ \eta} \Ex{x_1,\ldots,x_{\ell} \in
    X}{\rho_{A\to B}(y+x_1+\ldots+x_{\ell})} } 
\geq 1 - \delta/\eta \mper 
\]
\fi
\end{lemma}

\begin{proof} From Markov's inequality it follows that for at least a $(1-\delta/\eta)$-fraction of $y \in S$, the relation
\[\rho_{A\to B}(y)\approx_{\epsilon}\rho_{A\to B}(y+x_1+\ldots+x_{\ell})\]
holds for at least a $(1-\eta)$-fraction of $\ell$-tuples
$(x_1,\ldots,x_{\ell}) \in X^{\ell}$. Thus for at least a $(1-\delta/\eta)$-fraction of $y \in S$,
we have that
\ifconf
\begin{align*}
& \abs{\Ex{x_1,\ldots,x_{\ell} \in X}{\rho_{A\to B}(y+x_1+\ldots+x_{\ell})}  -  \rho_{A\to B}(y)} \\
& ~~\leq  \Ex{x_1,\ldots,x_{\ell} \in X}{\abs{\rho_{A\to B}(y+x_1+\ldots+x_{\ell})  - \rho_{A\to B}(y)}}
\end{align*}
\else
\[
\abs{\Ex{x_1,\ldots,x_{\ell} \in X}{\rho_{A\to B}(y+x_1+\ldots+x_{\ell})}  -  \rho_{A\to B}(y)} 
\leq  \Ex{x_1,\ldots,x_{\ell} \in X}{\abs{\rho_{A\to B}(y+x_1+\ldots+x_{\ell})  - \rho_{A\to B}(y)}}
\]
\fi
which is seen to be bounded above by $(1-\eta) \cdot \epsilon + \eta \cdot 1 \leq \epsilon + \eta$.
\end{proof}

The next lemma establishes (\ref{eq:subspace3}).

\begin{lemma}\label{lem:subspace-fourier}
Let $X \subseteq \F_2^n$, and let $V = \spec_{1/2}(X)^{\perp}$.
 Then for all $y \in \F_2^n$ and $v \in V$,
\ifconf
\begin{align}\label{eq:subspace4}
\nonumber &\Ex{x_1,\ldots,x_{\ell} \in X}{\rho_{A\to B}(y+x_1+\ldots+x_{\ell})} \\
&~~\approx_{\epsilon''} \Ex{x_1,\ldots,x_{\ell} \in X}{\rho_{A\to B}(y+v+x_1+\ldots+x_{\ell})}\mcom
\end{align}
\else
\begin{equation}\label{eq:subspace4}
\Ex{x_1,\ldots,x_{\ell} \in X}{\rho_{A\to B}(y+x_1+\ldots+x_{\ell})} 
\approx_{\epsilon''} \Ex{x_1,\ldots,x_{\ell} \in X}{\rho_{A\to B}(y+v+x_1+\ldots+x_{\ell})}\mcom
\end{equation}
\fi
where $\epsilon'' = 2^{-\ell} \sqrt{|B|/|A|}$.
\end{lemma}

\begin{proof} 
We can write the difference between the two sides of (\ref{eq:subspace4}) using the
convolution operator as
\[(\mu_{X})^{*\ell}*\mu_{A}*\indicator{B}(y)-(\mu_{X})^{*\ell}*\mu_{A}*\indicator{B}(y+v),\]
which in terms of the Fourier basis equals
\[
\sum_{t \in \F_2^n}  \hat \mu_{A}(t)  \cdot  (\hat \mu_{X}(t))^{\ell} \cdot \widehat
{\indicator{B}}  (t) \cdot 
\inparen{(-1)^{y \cdot t}- (-1)^{(y+v) \cdot t}} \mper
\]

This expression in turn is bounded in absolute value by
\[ 
\confomit{
\sum_{t \in \F_2^n} \abs{  \hat \mu_{A}(t) } \cdot \abs{ \hat \mu_{X}(t)}^{\ell} \cdot \abs{\widehat
  {\indicator{B}}  (t)} \cdot  \abs{(-1)^{y \cdot t}} \cdot \abs{1-(-1)^{v \cdot t}} =}
 \sum_{t \in \F_2^n} \abs{  \hat \mu_{A}(t) } \cdot \abs{ \hat \mu_{X}(t)}^{\ell} \cdot
 \abs{\widehat {\indicator{B}}  (t)}  \cdot \abs{1-(-1)^{v \cdot t}} \mper
\]
\ifconf
The above sum equals
$\sum_{t \notin V^{\perp}} \abs{  \hat \mu_{A}(t) } \cdot \abs{ \hat \mu_{X}(t)}^{\ell} \cdot
\abs{\widehat {\indicator{B}}  (t)}  \cdot \abs{1-(-1)^{v \cdot t}} 
$ which is at most $2^{-\ell}  \sum_{t \notin V^{\perp}} \abs{  \hat \mu_{A}(t) } \cdot
\abs{\widehat {\indicator{B}}  (t)}$ since $V$ is the orthogonal complement of $\spec_{1/2}(X)$.
By the Cauchy-Schwarz inequality and Parseval's indentity, this is bounded above by
\else
By definition of $V$ as the orthogonal complement of $\spec_{1/2}(X)$, the right-hand side can be
bounded as 
\[
\sum_{t \notin V^{\perp}} \abs{  \hat \mu_{A}(t) } \cdot \abs{ \hat \mu_{X}(t)}^{\ell} \cdot
\abs{\widehat {\indicator{B}}  (t)}  \cdot \abs{1-(-1)^{v \cdot t}} 
\leq   2^{-\ell}  \sum_{t \notin V^{\perp}} \abs{  \hat \mu_{A}(t) } \cdot \abs{\widehat {\indicator{B}}  (t)}\mper
\]
\fi
By the Cauchy-Schwarz inequality and Parseval's indentity, this is bounded above by
\ifconf
\begin{align*}
&2^{-\ell}  \sqrt{  \sum_{t \notin V^{\perp}} ( \hat \mu_{A}(t) )^2}
 \sqrt{\sum_{t \notin V^{\perp}} (\hat {\mathbf{1}_B}  (t))^2} \\
&\leq  2^{-\ell}  \sqrt{  \mathbb{E}_{y \in \F_2^n} (  \mu_{A}(y)) ^2}
 \sqrt{\mathbb{E}_{y \in  \F_2^n} (\mathbf{1}_B (y))^2} =
 2^{-\ell} \sqrt{|B|/|A|} \mper
\end{align*}
\else
\[
2^{-\ell}  \sqrt{  \sum_{t \notin V^{\perp}} ( \hat \mu_{A}(t) )^2}
 \sqrt{\sum_{t \notin V^{\perp}} (\hat {\mathbf{1}_B}  (t))^2} \leq  2^{-\ell}  \sqrt{  \mathbb{E}_{y \in \F_2^n} (  \mu_{A}(y)) ^2}
 \sqrt{\mathbb{E}_{y \in  \F_2^n} (\mathbf{1}_B (y))^2} =
 2^{-\ell} \sqrt{|B|/|A|} \mper
\]
\fi
\end{proof}

We are now ready for the proof of Corollary \ref{lem:ap-sumsets-subspace}.

\begin{proofof}{of Corollary \ref{lem:ap-sumsets-subspace}}
Let $X$ be the set guaranteed by Corollary \ref{lem:ap-sumsets-iterated} for $K = 1/\alpha$, and let $V =\spec_{1/2}(X)^{\perp}$.

First, note that Property 1 of Corollary \ref{lem:ap-sumsets-iterated} implies that $|X| \geq |A|/(2 (1/\alpha)^{t-1}) \geq \alpha^t \cdot 2^{n-1}$. It now follows from 
Chang's theorem (Theorem \ref{thm:chang}) that
\[\dim(\spec_{1/2}(X))= \codim(V) \leq 8\frac{\log(2/\alpha^t)} {(1/2)^2} = 32\log(2/\alpha^t).\]
It remains to show that (\ref{eq:subspace-main}) holds.

Let $\delta := 8 \ell (|A+B|/|S|)\cdot \exp\left(-2\eps^2 t\right)$.
From Corollary \ref{lem:ap-sumsets-iterated} and Lemma \ref{lem:subspace-averaging} we have that
\begin{equation}\label{eq:subspace5}
   \rho_{A\to B}(y)\approx_{2\epsilon \ell+ \eta} \mathbb{E}_{x_1,\ldots,x_{\ell} \in X}[\rho_{A\to B}(y+x_1+\ldots+x_{\ell})]
\end{equation}
for at least a $(1-\delta/\eta)$-fraction of $y \in S$, and similarly that for all $v\in V$,
\begin{equation}\label{eq:subspace6}
 \rho_{A\to B}(y+v)\approx_{2\epsilon \ell + \eta} \mathbb{E}_{x_1,\ldots,x_{\ell} \in X}[\rho_{A\to B}(y+v+x_1+\ldots+x_{\ell})]
 \end{equation}
 for at least a $(1-\delta/\eta)$-fraction of $y \in S$. Moreover, Lemma \ref{lem:subspace-fourier}
 implies that for every $y \in S$ and $v\in V$,
\ifconf
\begin{align}\label{eq:subspace7}
\nonumber &\Ex{x_1,\ldots,x_{\ell} \in X}{\rho_{A\to B}(y+x_1+\ldots+x_{\ell})} \\
&~~\approx_{2^{-\ell} \sqrt{|B/A|}} 
\Ex{x_1,\ldots,x_{\ell} \in X}{\rho_{A\to B}(y+v+x_1+\ldots+x_{\ell})} 
\mper
\end{align}
\else
\begin{equation}\label{eq:subspace7}
\Ex{x_1,\ldots,x_{\ell} \in X}{\rho_{A\to B}(y+x_1+\ldots+x_{\ell})} 
\approx_{2^{-\ell} \sqrt{|B/A|}} 
\Ex{x_1,\ldots,x_{\ell} \in X}{\rho_{A\to B}(y+v+x_1+\ldots+x_{\ell})} 
\mper
\end{equation}
\fi
Applying the union bound and the triangle inequality to  (\ref{eq:subspace5}), (\ref{eq:subspace6}) and (\ref{eq:subspace7}), we conclude that
\[ \rho_{A\to B}(y)\approx_{\epsilon'}  \rho_{A\to B}(y+v)\]
 for $\epsilon' = 4\epsilon \ell + 2 \eta + 2^{-\ell} \sqrt{|B|/|A|}$ for at least a ($1 -
 2\delta/\eta$)-fraction of $y \in S$, which is the desired conclusion.
\end{proofof}

\section{An improved quadratic Goldreich-Levin theorem}\label{sec:app1}

Both in number theory and theoretical computer science, there are certain situations where we
may wish to decompose a bounded function $f: \F_2^n \ra \C$ as a sum $g+h$, where $g$ is a
``uniform" or ``random-looking", and $h$ is a somewhat ``structured" part. Such situations include
the counting of arithmetic progressions \cite{GrML}, the analysis of Probabilistically Checkable Proofs (PCPs) \cite{SamorodnitskyT06} and the approximation of matrices and tensors \cite{FriezeK99}.

In the case where one is looking for ``linear uniformity" in the function $g$, for example when counting arithmetic progressions of length 3, such a decomposition is achieved by separating large and small Fourier coefficients (corresponding to ``linearly structured" and ``linearly uniform" parts, respectively). This task can be handled algorithmically by the Goldreich-Levin  theorem (\cite{GoldreichL89}, see Theorem \ref{thm:lgl} below), which provides an algorithm that computes, with high probability, the large Fourier coefficients of $f: \F_2^n \to \{-1,1\}$ in time polynomial in $n$.

\begin{theorem}[Goldreich-Levin Theorem]\label{thm:lgl}
Let $\nu,\delta>0$. There is a randomized algorithm which, given oracle access to a function $f: \F_2^n \to \pmone$, runs in time
$O(n^2\log n \cdot \poly(1/\nu, \log(1/\delta)))$ and outputs a decomposition
\[f = \sum_{i=1}^k
c_i \cdot (-1)^{\langle \alpha_i,x\rangle} + g\]
with the following guarantee.
\begin{itemize}
\item $k = O(1/\nu^2)$.
\item $\prob{\exists i~ \smallabs{c_i - \fhat(\alpha_i)} > \nu/2} \leq \delta$.
\item $\prob{\forall \alpha ~\text{such that}~
    \smallabs{\fhat(\alpha)} \geq \nu,~~ \exists i~ \alpha_i = \alpha} \geq 1-\delta$.
\end{itemize}
\end{theorem}

However, these linear decompositions have been shown to not be sensitive enough to handle many other situations, such as the counting of arithmetic progressions of length 4. In the latter case, one instead needs the function $g$ to be ``quadratically uniform" in the sense of Gowers \cite{GSzT4}. We say that a function $g$ is \emph{quadratically uniform} if it is small in the $U^3$ norm, which is defined by the formula
\[\|g\|_{U^3}^{8} = \E_{x,h_1, h_2, h_3 \in G}
\prod_{\omega \in \{0,1\}^3} C^{|\omega|}g(x+\omega\cdot h),\]
where $\omega\cdot h$ is shorthand for $\sum_i\omega_ih_i$, and $C^{|\omega|}g=g$ if $\sum_i\omega_i$ is even and $\overline{g}$ otherwise.

A hint as to what might constitute the \emph{quadratically structured} part of a decomposition in which $g$ is quadratically uniform is given by the so-called inverse theorem for the $U^3$ norm, whose proof was largely contained in Gowers's proof of Szemer\'edi's theorem
but brought to the point by Samorodnitsky \cite{Samorodnitsky07} (in the case of characteristic 2) and by Green and Tao \cite{GrTu3}.
It states, qualitatively speaking, that a function with
large $U^3$ norm correlates with a \emph{quadratic phase function}, by which we mean a function of
the form $(-1)^{q}$ for a quadratic form $q: \F_2^n \to \F_2$.

The inverse theorem implies that the structured part $h$ has quadratic structure in the
case where $g$ is small in $U^3$, and starting with \cite{GrML} a variety of such
\emph{quadratic decomposition theorems} have come into existence: in one formulation \cite{GW2},
one can write $f$ as
\begin{equation}\label{quaddecomp}
f=\sum_i \lambda_i (-1)^{q_i} + g+l,
\end{equation}
where the $q_i$ are quadratic forms, the $\lambda_i$ are real coefficients such that
$\sum_i |\lambda_i|$ is bounded, $\|g\|_{U^3}$ is small and $l$ is a small $\ell_1$
error (which is negligible in all known applications). Such a decomposition is not unique and non-trivial since the quadratic phases $\omega^q$, unlike linear exponentials, do not form an orthonormal basis. In analogy with the decomposition into Fourier characters, it is natural to think of the coefficients $\lambda_i$ as the \emph{quadratic Fourier coefficients} of $f$.

An algorithmic version of a \emph{quadratic} decomposition theorem was given by the last two authors in \cite{TulsianiW11}. Prior to \cite{TulsianiW11}, all quadratic decomposition theorems proved had been of a rather abstract nature. In particular, work by Trevisan, Vadhan and the third author \cite{TrevisanTV09} used
linear programming techniques and boosting, while Gowers and the last author \cite{GW2} gave a
(non-constructive) existence proof using the Hahn-Banach theorem. The main result of \cite{TulsianiW11} then was the following.

\begin{theorem}[Quadratic Goldreich-Levin theorem]\label{thm:decomposition-intro}
Let $\e, \delta > 0$, $n \in \N$ and $B > 1$. Then there exists $\eta = \exp((B/\e)^C)$ and a
randomized algorithm running in time $O(n^4 \log n \cdot \poly(1/\eta,\log(1/\delta)))$ which,
given any function $f: \F_2^n \to [-1,1]$ as an oracle, outputs with probability at least $1-\delta$ a
decomposition into quadratic phases
\[ f~=~ c_1 (-1)^{q_1} + \ldots + c_k (-1)^{q_k} + g + l \]
satisfying $k \leq 1/\eta^2$,   $\uthreenorm{g} \leq \epsilon$, $\norm{l}_1 \leq 1/2B$ and $|c_i| \leq \eta$
for all $i=1,\dots, k$.
\end{theorem}

The algorithm comprised two parts. The first was a (entirely deterministic) procedure for assembling the quadratic phases with which the function $f$ correlates into an actual decomposition, if these quadratic phases can indeed be found.

\begin{theorem}\label{thm:decomposition-general}
Let $\cal Q$ be a class of functions as above and let $\e, \delta > 0$ and $B > 1$.
Let $A$ be an algorithm which, given oracle access to a function
$f: X \to [-B,B]$ satisfying $\snorm{f} \geq \e$, outputs, with probability at least $1-\delta$, a function
$\q \in {\cal Q}$ such that $\ip{f,\q} \geq \eta$ for some $\eta = \eta(\e,B)$. Then there exists an
algorithm which, given any
function $f: X \to [-1,1]$, outputs with probability at least $1-\delta/\eta^2$ a
decomposition
\[ f~=~ c_1 \q_1 + \ldots + c_k \q_k + g + l \]
satisfying $k \leq 1/\eta^2$,   $\snorm{g} \leq \epsilon$, $\norm{l}_1 \leq 1/2B$ and $|c_i| \leq \eta$
for all $i=1,\dots, k$.

The algorithm makes at most $k$ calls to $A$.
\end{theorem}

Theorem \ref{thm:decomposition-general} is proved using a boosting argument, for which we refer the reader to \cite{TulsianiW11}.
The other key component in the quadratic Goldreich-Levin algorithm was the following self-correction procedure for Reed-Muller codes of order 2 (which are simply truth-tables of quadratic phase functions).
\begin{theorem}\label{thm:FindQuadratic-Intro}
Given $\e, \delta > 0$, there exists $\eta = \exp(-1/\e^{C})$ and
a randomized algorithm \FindQuadratic running in time $O(n^4 \log n \cdot \poly(1/\e, 1/\eta,
\log(1/\delta)))$ which,
given oracle access to a function $f: \F_2^n \to \pmone$, either outputs a quadratic form $q(x)$
or $\bot$. The algorithm satisfies the following guarantee.
\begin{itemize}
\item If $\uthreenorm{f} \geq \e$, then with probability at
  least $1-\delta$ it finds a quadratic form $q$ such that $\ip{f,(-1)^q} \geq \eta$.
\item The probability that the algorithm outputs a quadratic form $q$ with $\ip{f,(-1)^q} \leq \eta/2$ is at
most $\delta$.
\end{itemize}
\end{theorem}

This is essentially an algorithmic version of the $U^3$ inverse theorem. The proof of Theorem \ref{thm:FindQuadratic-Intro} follows that of the
inverse theorem very closely, except that many of the results from additive combinatorics that are
used in the process need to be replaced by new ``sampling versions": since the subsets of $\F_2^n$
that appear in the proof are generally very dense, it is too expensive to even write them down (let
alone perform operations on them) if one is aiming for an algorithm that runs in time polynomial in
$n$.

A crucial ingredient in the proof of Theorem \ref{thm:FindQuadratic-Intro} was an algorithmic
version of the Bogolyubov-Ruzsa lemma (Lemma 5.3 in \cite{TulsianiW11}), which reads as follows.

\ifconf
\begin{lemma}\label{lem:bogolyubov}
\else
\begin{lemma}[Algorithmic Bogolyubov-Ruzsa Lemma]\label{lem:bogolyubov}
\fi
There exists a randomized algorithm \bogolyubov with parameters $\rho$ and $\delta$ which, given
oracle access to a function $h: \F_2^n \ra \{0,1\}$ with $\E h \geq \alpha$, outputs a subspace $V
\leqslant \F_{2}^{n}$ of codimension at most $O(\alpha^{-3})$ (by giving a basis for $\ortho{V}$) such
that with probability at least $1-\delta$, we have $h*h*h*h(x)>\rho^4/2$ for all  $x \in V$. The
algorithm runs in time $n^2 \log n\cdot \poly(1/\alpha, \log(1/\d))$.
\end{lemma}

In Section \ref{subsec:bralgo} we develop a replacement for this lemma (Theorem
\ref{thm:bogolyubov-algo} below) with much better bounds. Inserting it into the framework of
\cite{TulsianiW11} reduces the dependence on $\eps$ in the running time and the number of terms in
the decomposition to quasipolynomial, allowing us to state the following result.

\ifconf
\begin{theorem}\label{thm:decomposition-intro}
\else
\begin{theorem}[Quasipolynomial Quadratic Goldreich-Levin]\label{thm:decomposition-intro}
\fi
Let $\e, \delta > 0$, $n \in \N$ and $B > 1$. Then there exists $\eta = \exp(\poly(B,\log(1/\eps)))$ and a
randomized algorithm running in time $O(n^4 \log n \cdot \poly(1/\eta,\log(1/\delta)))$ which,
given any function $f: \F_2^n \to [-1,1]$ as an oracle, outputs with probability at least $1-\delta$ a
decomposition into quadratic phases
\[ f~=~ c_1 (-1)^{q_1} + \ldots + c_k (-1)^{q_k} + g + l \]
satisfying $k \leq 1/\eta^2$,   $\uthreenorm{g} \leq \epsilon$, $\norm{l}_1 \leq 1/2B$ and $|c_i| \leq \eta$
for all $i=1,\dots, k$.
\end{theorem}

A further variant of Theorem \ref{thm:decomposition-intro} was proved in \cite{TulsianiW11}, in
which the quadratic phases in the decomposition were replaced with slightly more complicated
quadratic object, namely so-called \emph{quadratic averages}, which were first introduced in
\cite{GW2} by Gowers and the last author. In this case the authors of \cite{TulsianiW11} obtained a
bound on the number of terms in the decomposition that was polynomial in $\eps^{-1}$ in time exponential in $\eps^{-1}$, at the cost of
the description size of each quadratic average being exponential in $\eps^{-1}$. Inserting the
Quasipolynomial Algorithmic Bogolyubov-Ruzsa Lemma (Theorem \ref{thm:bogolyubov-algo}) in the work of Section
5 in \cite{TulsianiW11}, we obtain an algorithm which finds a decomposition into polynomially many
quadratic averages in time quasipolynomial in $\eps^{-1}$, where the description size of each average is
now quasipolynomial in $\eps^{-1}$. We leave the details to the interested reader.

\subsection{The quasipolynomial Bogolyubov-Ruzsa lemma}\label{subsec:strong-bogol}

In the context of $\F_2^n$, the traditional Bogolyubov-Ruzsa lemma states that if a set $A$ has density at least $\alpha$ in its ambient group, then its fourfold sumset $A+A+A+A$ contains a subspace of codimension at most $2\alpha^{-2}$. It is easily proved using a few lines of Fourier analysis: the orthogonal complement of the subspace is given by the frequencies at which the indicator function of $A$ has relatively large Fourier coefficients.

The bound on the codimension of $V$ was improved to $O(\log^4( \alpha^{-1}))$ by Sanders
\cite{Sanders10}. This improvement has far-reaching quantitative implications for other problems, in particular
the bound in Roth's theorem \cite{Sanders:roth} and the $U^3$ inverse theorem. We now deduce the quasipolynomial
Bogloyubov-Ruzsa lemma from Corollary \ref{lem:ap-sumsets-subspace}. In Section \ref{subsec:bralgo} we give an algorithmic version of the proof, which allows us to explicitly find a basis for $V^\perp$.

\ifconf
\begin{theorem}\label{thm:strong-bogolyubov}
\else
\begin{theorem}[Quasipolynomial Bogloyubov-Ruzsa Lemma]\label{thm:strong-bogolyubov}
\fi
Let $A \subseteq \F_2^n$ be a subset of density $\alpha$. Then there exists a subspace $V$ of
$\F_2^n$ satisfying $V \subseteq 4A$ and
\confequation{\codim(V) = O(\log^4( \alpha^{-1})).}
\end{theorem}
\begin{proofof}{of Theorem \ref{thm:strong-bogolyubov}}
Applying Corollary \ref{lem:ap-sumsets-subspace} with $B = 2A$, $S = A$, $\ell = \log(30^2/\alpha)/2$, $\eta = 1/60$,
$\epsilon = 1/(120 \ell)$ and $t$ to be determined later on, we conclude the existence of
a subspace $V$ of $\codim(V) \leq 32 \log (2/\alpha^t)$ which has the property that for all $v \in V$,
\[ \Pr_{a\in A}\left[\rho_{A\to 2A}(a)\approx_{\epsilon'}\rho_{A\to 2A}(a+v)\right]
    \geq  1-16 \frac{\ell} {\eta} \frac{|3A|}{|A|}\cdot \exp\left(-2\eps^2 t\right),\]
where $\epsilon' = 4\eps \ell + 2 \eta +2^{-\ell} \sqrt{|2A|/|A|} \leq  1/30 + 1/30 + (\sqrt{\alpha}/30) \cdot \sqrt{1/\alpha} \leq 1/10$.

Since $\rho_{A \to 2A}(a) = 1 $ for all $a \in A$, this implies that
\ifconf
\begin{align*}
 &\Pr_{a\in A}\left[\rho_{A\to 2A}(a+v) \geq 0.9 \right] \\
 &\quad=~       \Pr_{a\in A}\left[\rho_{A\to 2A}(a)\approx_{1/10}\rho_{A\to 2A}(a+v)\right] \\
 &\quad\geq~  1-16 \frac{\ell} {\eta} \frac{|3A|}{|A|}\cdot \exp\left(-2\eps^2 t\right) \\
 &\quad\geq~  1-16 \frac{\ell} {\eta} \alpha^{-1} \cdot \exp\left(-2\eps^2 t\right)  \\ 
 &\quad=~       1- 480 \frac{\log(30^2 /\alpha)} {\alpha} \cdot \exp\left(-\frac{t} {1800 \cdot \log^2(30^2/\alpha)}\right)
\end{align*}
\else
\begin{eqnarray*}
 \Pr_{a\in A}\left[\rho_{A\to 2A}(a+v) \geq 0.9 \right] & = &
    \Pr_{a\in A}\left[\rho_{A\to 2A}(a)\approx_{1/10}\rho_{A\to 2A}(a+v)\right] \\
    &\geq & 1-16 \frac{\ell} {\eta} \frac{|3A|}{|A|}\cdot \exp\left(-2\eps^2 t\right) \\
     & \geq & 1-16 \frac{\ell} {\eta} \alpha^{-1} \cdot \exp\left(-2\eps^2 t\right)  \\ 
     & = & 1- 480 \frac{\log(30^2 /\alpha)} {\alpha} \cdot \exp\left(-\frac{t} {1800 \cdot \log^2(30^2/\alpha)}\right)
\end{eqnarray*}
\fi

Choosing $t = O(\log^3(1/\alpha))$ we find that $V$ has the desired codimension $O(\log^4(1/\alpha))$ and satisfies
\[\Pr_{a\in A}\left[\rho_{A\to 2A}(a+v) \geq 0.9 \right] \geq 0.9\]
for all $v \in V$. Recalling the definition of $\rho_{A \to B}$ in (\ref{eq:mu-definition}), this inequality implies that for all $v \in V$,
\[\Pr_{a, a' \in A} [a+a'+v \in 2A] \geq 0.9^2 = 0.81.\]
By averaging, there therefore exists a pair $a, a' \in A$ such that $\Pr_{v \in V}[a+a'+v \in 2A] \geq 0.81$, or equivalently
$|V \cap (a+a'+2A)\mid \geq 0.81 |V|$. But it is easy to see that if $|V \cap B| > \frac {1} {2} |V|$ for some subset $B \subseteq \F_2^n$, then $V \subseteq 2B$ (since every element $v \in V$ has precisely $|V|$ different representations as $v=v_1+v_2$ where $v_1,v_2 \in V$). We conclude that $V \subseteq 2(a+a'+A+A) \subseteq 4A$, which finishes the proof.
\end{proofof}

\subsection{The quasipolynomial Bogolyubov-Ruzsa lemma -- algorithmic version}\label{subsec:bralgo}

Here we develop an algorithmic version of the quasipolynomial Bogolyubov-Ruzsa lemma. In other words, we give
an efficient (probabilistic) algorithm for finding a basis for the orthogonal complement of the
subspace $V$ in Theorem \ref{thm:strong-bogolyubov}.

\ifconf
\begin{theorem}\label{thm:bogolyubov-algo}
\else
\begin{theorem}[Algorithmic Quasipolynomial Bogolyubov-Ruzsa Lemma]\label{thm:bogolyubov-algo}
\fi
There exists a randomized algorithm \strongbogolyubov with parameters $\alpha$ and $\gamma'$ which, given
oracle access to a function $h: \F_2^n \ra \{0,1\}$ with $\E h \geq \alpha$, outputs a subspace $V
\leqslant \F_{2}^{n}$ of codimension at most $O(\log^4(1/\alpha))$ (by giving a basis for $\ortho{V}$) such
that with probability at least $1-\gamma'$, we have $h*h*h*h(x)>0$ for each $x \in V$. The
algorithm runs in time $2^{O(\log^4(1/\alpha))} \cdot \polylog(1/\gamma') \cdot n^3 \log n$.
\end{theorem}

Note that if the function $h$ equals the indicator function of a subset $A \subseteq \F_2^n$, then the condition $h*h*h*h(x)>0$ implies that $x \in 4A$, and if this condition is satisfied for all $x \in V$, then $V \subseteq 4A$. However, while it will be convenient to think of the set $A = \{x  \in \F_2^n ~|~ h(x) = 1\}$ in the proof, we will actually apply the
theorem to the output of a randomized algorithm, for which the statement in terms of a function makes
more sense. We also assume for convenience that $\E h$ is \emph{exactly} $\alpha$.
The proof remains unchanged when the density is larger than $\alpha$.

In the combinatorial proof we considered the measure $\rho_{A \to 2A}(y) = \Prob{a \in A}{y + a \in
  2A}$, and the subspace $V$ was defined in terms of a set $X$ which was described using this measure. However,
now this measure is difficult to compute since it might not be possible to test membership simply
using oracle access to $h$, which is the indicating function for $A$. We give a robust version of
the combinatorial proof by noting that $y + a \in 2A$ is equivalent to saying that $h * h(a+y) >
0$. But since we do not have noise-free access to $h*h$, we cannot test this function directly. Instead, we test if $h * h (a + y) \geq \eta \alpha^2$ for some $\eta > 0$. For this purpose, we define the set
\[ Z_{\eta} ~:=~ \condset{x \in \F_2^n}{h * h(x) \geq \eta \cdot \alpha^2} \mper \]
The following procedure tests membership in $Z_{\eta}$ by estimating $h * h$ using few samples.
\fbox{
\begin{minipage}{0.9\columnwidth}

\smallskip

\Ztest($x$)
\begin{itemize}
\item[-] Estimate the expectation $h*h(x)=\mathbb{E}_{y \in \F_2^n} h(y)\cdot h(x-y) $ using $r$ samples of elements $y \in \F_2^n$.
\item[-] Answer 1 if the estimate is at least $\eta \alpha^2$ and 0 otherwise.
\end{itemize}

\end{minipage}
}
\smallskip

However, since we are \emph{estimating} the value of $h*h$, we only have the following kind of
guarantee.
%
\begin{claim}\label{clm:Ztest}
Given $\gamma_1 > 0$, the output of ~$\Ztest(x)$ with
$r = O(1/(\eta^2 \alpha^4) \cdot \log(1/\gamma_1))$ queries satisfies the following
guarantee with probability at least $1-\gamma_1$.
\begin{itemize}
\item $\Ztest(x) = 1  ~\Longrightarrow x \in Z_{\eta/2}$.
\item $\Ztest(x) = 0  ~\Longrightarrow x \notin Z_{3\eta/2}$.
\end{itemize}
\end{claim}
\begin{proof} This follows immediately from the Hoeffding bound (Lemma \ref{lem:hoeffding-sample}).
\end{proof}

Let $Z(x)$ denote the (random) function given by  the output of $\Ztest(x)$. Then the measure
\[
\rr(y)    ~:=~ \Ex{a \in A}{Z(y + a)} ~=~ \frac{1}{\alpha} \cdot h * Z (y) \mper
\]
can be efficiently estimated by sampling.

Next we need a procedure to test for membership in the set $X$ which satisfies the iterated almost-periodicity condition in Corollary \ref{lem:ap-sumsets-iterated}.

\begin{lemma}\label{lem:xtest}
Let $\eta \in (0,1/3600)$, $\gamma_2 > 0$ and let $h$ and $A$ be defined as above.
Then for any integers $\ell$, $t=O(\ell^2 (\log \ell +  \log(1/\alpha)))$ there exists a randomized procedure $\Xtest$ with outputs in $\B$ which runs in time $(1/\alpha)^{O(t)}$ and has the following properties.
\begin{itemize}
\item With probability $1-\gamma_2$, $\Prob{x \in \F_2^n}{\Xtest(x) = 1} ~\geq~ \alpha^{2t}/4$.
\item For all $x_1,\ldots,x_{\ell} \in \F_2^n$, we have with probability at least $1-\gamma_2$,
\ifconf
that if $\Xtest(x_i) =  1~\forall i \in [\ell]$ then $\Ex{a,a' \in
  A}{Z(a+a'+x_1+\ldots+x_{\ell})} \geq \frac{9}{10}$
\else
\[\forall i \in [\ell] ~\Xtest(x_i) =  1 ~~~\Longrightarrow~~~ \Ex{a,a' \in
  A}{Z(a+a'+x_1+\ldots+x_{\ell})} \geq \frac{9}{10} \mper\]
\fi
\end{itemize}
\end{lemma}
\begin{proof}
As in the proof of Proposition~\ref{prop:ap-sumsets}, define $\Gnew$ to be the set of sequences $\bfa \in (\F_2^n)^t$ which can be used to estimate
$\rho$ well (for our new definition of $\rho$). For $\bfa \in A^t$, define
\begin{align*}
\rra(y)    &~:=~ \frac{1}{t} \cdot \sum_{i=1}^{t} {Z(a_i+y)} \mcom \\
\Gnew &~:=~ \condset{\bfa \in A^t}{\Prob{y \in \F_2^n}{{\mid \rr(y) - \rra(y)\mid} \geq \e} \leq \delta} \mper
\end{align*}
Also, as above, we will need to test membership in $\Gnew$. This we will only be able to do
approximately, using the following randomized procedure.

\fbox{
\begin{minipage}{0.9\columnwidth}

\smallskip

\Gtest($\bfa = (a_1, \ldots,a_t)$)
\begin{itemize}
\item[-] Check if $h(a_1) = \ldots = h(a_t) = 1$. If not output 0.
\item[-] Pick $r$ independent samples $y_1, \ldots, y_r \in \F_2^n$.
\item[-] For each $y_i$, estimate $\rr(y_i)$ using $r'$ independent samples. Also compute $\rra(y_i)$
  for each $y_i$.
\item[-] If $\abs{\condset{y_i}{\abs{\rr(y_i) - \rra(y_i)} \geq \eps}} > \delta r$ then output 0, else
  output 1.
\end{itemize}

\end{minipage}
}
\smallskip

We prove the following guarantee for the above test.
\begin{claim}\label{clm:Gtest}
Given $\gamma_3 > 0$, the output of ~$\Gtest(\bfa)$ with
$r = O(1/\delta^2 \cdot \log(1/\gamma_3))$ and $r' = O(1/\eps^2 \cdot \log(r/\gamma_3))$ queries,
satisfies the following guarantee with probability at least $1-\gamma_3$.
\begin{itemize}
\item $\Gtest(\bfa) = 1  ~\Longrightarrow \bfa \in \Gparam{2\eps,2\delta}$.
\item $\Gtest(\bfa) = 0  ~\Longrightarrow \bfa \notin \Gparam{\eps/2, \delta/2}$.
\end{itemize}
\end{claim}
\begin{proof} Again, this is a direct consequence of the Hoeffding bound (Lemma \ref{lem:hoeffding-sample}).
\end{proof}
Note that the definition of the procedure $\Gtest$ actually depends on the parameters $\eps,\delta$ and the error
parameter $\gamma_3$, for choosing the appropriate values of $r$ and $r'$. However, we choose to
hide this dependence for the sake of readability.

From now on let $G(\bfa)$ denote the output of $\Gtest$ on the input $\bfa$. We will now find an element $\ha \in (\F_2^n)^t$ such that $G(\ha) = 1$ and $G(\ha+x) = 1$ for a large number of
elements $x \in \F_2^n$. This can be done efficiently with high probability it $\delta = \exp(-O(\epsilon^2 t))$.
\begin{claim}\label{clm:find-ahat}
Given $\gamma_4 > 0$, there exists an algorithm which makes $O((1/\alpha^{6t}) \cdot
\log^2(1/\gamma_4))$ calls to $\Gtest$ and finds an $\ha \in (\F_2^n)^t$ such that with probability
$1-\gamma_4$, we have $G(\ha) = 1$ and
$\Ex{x \in \F_2^n}{G(\ha + x)} \geq \alpha^{2t}/4$.
\end{claim}
\begin{proof}
The Hoeffding bound gives that $\abs{\Gparam{\eps/2,\delta/2}} \geq 0.99 \abs{A^t}$ for
$t \geq (c/\e^2)\cdot \log(1/\delta)$. Since $A$ has density $\alpha$ in $\F_2^n$ we have for $\gamma_3 <  0.04$ that
$\Ex{\bfa \in (\F_2^n)^t}{G(\bfa)} \geq (1-\gamma_3) \cdot (0.99 \alpha^t) \geq 0.95 \alpha^t$.
Using convexity gives
\ifconf
\begin{align*}
&\Ex{\bfa \in (\F_2^n)^t, x \in \F_2^n}{G(\bfa) \cdot G(\bfa+x)}\\
&\quad=~ \Ex{\bfa \in (\F_2^n)^t, x,x' \in \F_2^n}{G(\bfa+x) \cdot G(\bfa+x')} \\
&\quad=~ \Ex{\bfa \in (\F_2^n)^t}{\inparen{\Ex{x \in \F_2^n}{G(\bfa+x)}}^2} \\
&\quad\geq~ \inparen{\Ex{\bfa \in (\F_2^n)^t, x\in \F_2^n}{G(\bfa+x)}}^2 \\
&\quad\geq~ (0.95 \cdot \alpha^t)^2 ~\geq~ 0.9 \cdot \alpha^{2t} \mper
\end{align*}
\else
\begin{align*}
\Ex{\bfa \in (\F_2^n)^t, x \in \F_2^n}{G(\bfa) \cdot G(\bfa+x)}
&~=~ \Ex{\bfa \in (\F_2^n)^t, x,x' \in \F_2^n}{G(\bfa+x) \cdot G(\bfa+x')} \\
&~=~ \Ex{\bfa \in (\F_2^n)^t}{\inparen{\Ex{x \in \F_2^n}{G(\bfa+x)}}^2} \\
&~\geq~ \inparen{\Ex{\bfa \in (\F_2^n)^t, x\in \F_2^n}{G(\bfa+x)}}^2 \\
&~\geq~ (0.95 \cdot \alpha^t)^2 ~\geq~ 0.9 \cdot \alpha^{2t} \mper
\end{align*}
\fi
Hence, by Markov's inequality
\[
\Prob{\bfa \in (\F_2^n)^t}{G(\bfa) \cdot \Ex{x \in \F_2^n}{G(\bfa + x)}
\geq \alpha^{2t}/2} \geq \alpha^{2t}/4 \mper
\]
The algorithm then simply tries random sequences $\bfa$ until it finds one for which $G(\bfa) = 1$. For such
an $\bfa$, it estimates $\Ex{x \in \F_2^n}{G(\bfa + x)}$ using $O((1/\alpha^{4t}) \cdot
\log(1/\gamma_4))$ samples. With probabilbity $1-\gamma_4/2$, the estimate is accurate to within an
additive $\alpha^{2t}/8$. The algorithm stops and outputs an $\ha$ for which $G(\ha) = 1$ and the
estimate computed by the algorithm is at least $3\alpha^{2t}/4$. By the above, it finds such an
$\ha$ with probability at least $1-\gamma_4/2$ in at most $O((1/\alpha^{2t})\cdot \log(1/\gamma_4))$
attempts. If not, it simply outputs a random $\ha$.
\end{proof}
Given $\ha$ as above, we define $X$ to be the set
\[
X ~:=~ \condset{x \in \F_2^n}{G(\ha + x) = 1} \mper
\]
Note that $G(\ha)=1$ and $|X| \geq (\alpha^{2t}/4) \cdot 2^n$ with probability $1-\gamma_4$.
Also, membership in $X$ can be tested efficiently. We simply define the procedure $\Xtest$ as
\[ \Xtest(x) = \Gtest(\ha+x) \mper\]
We now prove that this $X$ suffices for our purposes. We will prove using induction that
for $x_1, \ldots, x_{\ell}$ satisfying $\Xtest(x_i) = 1 ~\forall i \in [\ell]$, we have with
probability $1-\gamma_3$ that
\ifconf
\begin{align*}
&\Ex{a,a' \in A}{Z(a+a'+x_1+\ldots+x_{\ell})} \\
&\quad=~\Ex{a \in A}{\rho(a+x_1+\ldots+x_{\ell})} ~\geq~ 9/10 \mper
\end{align*}
\else
\[ \Ex{a,a' \in A}{Z(a+a'+x_1+\ldots+x_{\ell})} ~=~
\Ex{a \in A}{\rho(a+x_1+\ldots+x_{\ell})} ~\geq~ 9/10 \mper\]
\fi
By Claims \ref{clm:Gtest} and \ref{clm:find-ahat}
we have that
$\ha, \ha+x_1,\ldots, \ldots, \ha+x_{\ell} \in \Gparam{2\e,2\delta}$
with probability at least
$1 - (\ell+1)\gamma_3 -\gamma_4$.
We will prove that whenever $\ha, x_1,\ldots,x_{\ell}$ satisfy this condition, then for all $r \in \{0,\ldots,\ell\}$
we have
\ifconf
\begin{align*}
&\Prob{a \in A}{\rho(a+x_1 + \ldots + x_r) \geq 1 - \sqrt{\eta + \gamma_1/\alpha^2} - 4r\e}\\
&\quad\geq~ 1-\sqrt{\eta + \gamma_1/\alpha^2} - 4\delta r /\alpha\mcom
\end{align*}
\else
\[ \Prob{a \in A}{\rho(a+x_1 + \ldots + x_r) \geq 1 - \sqrt{\eta + \gamma_1/\alpha^2} - 4r\e}
~\geq~ 1-\sqrt{\eta + \gamma_1/\alpha^2} - 4\delta r/\alpha\mcom\]
\fi
where $\gamma_1$ is the error parameter in Claim \ref{clm:Ztest}. The following claim proves the
base case $r = 0$.
\begin{claim}\label{claim:large-mu}
$\Prob{a \in A}{\rho(a) \geq 1-\sqrt{\eta + (\gamma_1/\alpha^2)}} ~\geq~ 1-\sqrt{\eta+(\gamma_1/\alpha^2)}$.
\end{claim}
\begin{proof}
We have
\ifconf
\begin{align*}
\Ex{a \in A}{\rho(a)} &~=~ \Ex{a,a' \in A}{Z(a+a')} \\
&~=~ \frac{1}{\alpha^2} \cdot \ip{h * h, Z} \\
&~=~ \frac{1}{\alpha^2} \cdot \ip{h * h, 1} - \frac{1}{\alpha^2} \cdot \ip{h * h, (1-Z)} \\
&~=~ 1 - \frac{1}{\alpha^2} \cdot \ip{h * h, (1-Z)} \
\end{align*}
\else
\begin{align*}
\Ex{a \in A}{\rho(a)} ~=~ \Ex{a,a' \in A}{Z(a+a')}
&~=~ \frac{1}{\alpha^2} \cdot \ip{h * h, Z} \\
&~=~ \frac{1}{\alpha^2} \cdot \ip{h * h, 1} - \frac{1}{\alpha^2} \cdot \ip{h* h, (1-Z)} \\
&~=~ 1 - \frac{1}{\alpha^2} \cdot \ip{h * h, (1-Z)} \
\end{align*}
\fi
Since with probability at least $1-\gamma_1$, $h*h$ is at most $\eta \alpha^2$ when $1-Z = 1$, the
inner product in the second term is at most $\eta \alpha^2 + \gamma_1$. This gives
$\Ex{a \in A}{\rho(a)} \geq 1 - \eta - (\gamma_1/\alpha^2)$. An averaging argument then proves the claim.
\end{proof}
For readability, let $\beta$ denote the quantity $\sqrt{\eta + (\gamma_1/\alpha^2)}$.
We assume by induction that for $x_1,\ldots,x_r$
\[ \Prob{a \in A}{\rr(a+x_1+\ldots+x_r) \geq 1-\beta - 4r\e}
~\geq~ 1-\beta-4r\delta/\alpha \mper\]
Define the set 
\[A_r = \condset{y = a+x_1+\ldots+x_r}{a \in A, \rr(y) \geq 1-\beta-4r\e} \mper\] 
By the above, the
density of $A_r$ is at least $\alpha_r = \alpha \cdot (1-\beta-4r\delta/\alpha)$. For any $x_{r+1}
\in X$, we have that for at least a $(1-4\delta/\alpha_r)$-fraction of elements $y \in A_r$,
\ifconf
$\abs{\rr(y) - \rraa{\ha}(y)} \leq \e$
and
$\abs{\rr(y+x_{r+1}) - \rraa{\ha+x_{r+1}}(y+x_{r+1})} \leq \e$.
\else
\[
\abs{\rr(y) - \rraa{\ha}(y)} \leq \e
\quad\text{and}\quad
\abs{\rr(y+x_{r+1}) - \rraa{\ha+x_{r+1}}(y+x_{r+1})} = \abs{\rr(y+x_{r+1}) - \rraa{\ha}(y)}\leq \e \mper
\]
\fi
Thus, by the triangle inequality
\ifconf
\begin{align*}
&\Prob{a \in A}{\rr(a+x_1+\ldots+x_{r+1}) \geq 1-\beta-4(r+1)\e}\\
&\quad\geq~ (1-4\delta/\alpha_r) \cdot (\alpha_r/\alpha) 
~=~ 1 - \beta - 4(r+1)\delta/\alpha \mper
\end{align*}
\else
\[
\Prob{a \in A}{\rr(a+x_1+\ldots+x_{r+1}) \geq 1-\beta-4(r+1)\e}
~\geq~ (1-4\delta/\alpha_r) \cdot (\alpha_r/\alpha) ~=~ 1 - \beta - 4(r+1)\delta/\alpha \mper
\]
\fi
To get the required bounds we choose $\gamma_1 = \eta \alpha^2$, $\gamma_3 = \gamma_2/(2(\ell+1))$ and
$\gamma_4 = \gamma_2/2$. Also, we take $\e = 1/(120 \ell), \delta = \alpha/(120 \ell)$ and
$\eta = 10^{-4}$.

The basic procedure used for the above algorithm is $\Ztest$, for which the running time is dominated
by $O((1/\alpha^4) \log(1/\alpha))$ oracle queries to the function $h$. Assuming the query can be
answered in constant time and it takes $O(n)$ time to write down the input, the running time for
$\Ztest$ is $O((1/\alpha^4) \log(1/\alpha) \cdot n)$. Also, for the above choice of parameters, the
procedure $\Gtest$ makes $O((\ell^4/\alpha^2) \cdot \log(\ell/\alpha) \cdot \log^3(\ell/\gamma_2))$
calls to $\Ztest$. Finally, the procedure in Claim \ref{clm:find-ahat} makes $O((1/\alpha^{6t})
\cdot \log^2(\gamma_2))$ calls to $\Gtest$. Taking $t$ to be $O((1/\eps^2) \log (1/\delta)) =
O(\ell^2 \cdot \log(\ell/\alpha))$, this gives a running time of $(1/\alpha)^{O(\ell^2 (\log \ell +
  \log(1/\alpha)))} \cdot \log^5(1/\gamma_2) \cdot n$ for the algorithm.
\end{proof}

We can now prove following algorithmic analogue of Theorem \ref{thm:strong-bogolyubov}.

\begin{proofof}{of Theorem \ref{thm:bogolyubov-algo}}
Choose $\ell = \log(10/\alpha)$. Let $X$ be the set defined in Lemma \ref{lem:xtest}.
Define the subspace $V_0$ as
\[ V_0 ~:=~ \ortho{\condset{\zeta \in \F_2^n}{\abs{\widehat{\oX}(\zeta)} \geq \widehat{\oX}(0)/2 }} = \big(\text{Spec}_{1/2}(X)\big)^{\perp}\]
where $\oX$ is the indicator function of $X$. To find (an approximation to) $V_0$, we first
estimate $\widehat{\oX}(0) = \ex{\oX} = \Prob{x \in \F_2^n}{\Xtest(x) = 1}$ using
$\Omega((1/\alpha^{4t}) \cdot \log(1/\gamma'))$ samples so that with probability $1-\gamma'/4$, the
error is at most $\alpha^{2t}/8$. By Lemma \ref{lem:xtest}, with probability $1-\gamma_2$, the
quantity $\Prob{x \in \F_2^n}{\Xtest(x) = 1}$ is at least $\alpha^{2t}/4$. Taking $\gamma_2 =
\gamma'/4$, we get that with probability $1-\gamma'/2$, the estimate is at least
$3\alpha^{2t}/8$. Call this estimate $\mu_0$.

We now need a procedure which determines the large Fourier coefficients of $\indicator{X}$ with
reasonable accuracy. This procedure is given by the Goldreich-Levin theorem (Theorem \ref{thm:lgl}).

We run Theorem \ref{thm:lgl} with error parameter
$\delta = \gamma'/2$ and an oracle access to the procedure $\Xtest$, to find all characters with Fourier
coefficients larger than $\mu_0/8$ in absolute value, up to an additive accuracy of $\nu = \mu_0/16$.
Let $K$ be the list of characters given by the algorithm. We take
\[V = \ortho{\condset{\zeta \in \F_2^n}{ \zeta \in K}} \mper\]

Now with probability at least $1-\gamma'$, the trivial coefficient $\widehat{\oX}(0)$ is at least $\alpha^{2t}/4$,
$K$ contains all $\zeta$ such that $|\widehat{\oX}(\zeta)|\geq \widehat{\oX}(0)/2$ and
$|\widehat{\oX}(\zeta)| \geq \widehat{\oX}(0)/32$ for all $\zeta \in K$.
By Chang's theorem (Theorem \ref{thm:chang}) and our choice of parameters, the codimension of $V$ is
then at most $O(\log(1/\alpha^{2t})) = O(\log^4(1/\alpha))$. It remains to show that $h*h*h*h(x) >
0$ for all $x \in V$.

Let $g =\indicator{Z_{\eta}} * \mu_A * \mu_X * \ldots * \mu_X$, where $\mu_X$ is convolved $\ell$ times. By definition of $X$, we have that
\[ \ip{g, \mu_A} ~\geq~ 9/10 \mper\]
Also, by definition of $V$, we have
\ifconf
\begin{align*}
&\abs{ \ip{g,\mu_A} - \ip{g*\mu_V, \mu_A} } \\
&\quad \leq~ \sum_{\gamma \notin \ortho{V}}  \abs{\widehat{\oZ}(\gamma)}
\abs{\widehat{\mu_A}(\gamma)}^2 \abs{\widehat{\mu_X}(\gamma)}^{\ell} 
~\leq~ \frac{1}{10}\mper
\end{align*}
\else
\[ 
\abs{ \ip{g,\mu_A} - \ip{g*\mu_V, \mu_A} }~ \leq~ \sum_{\gamma \notin \ortho{V}}  \abs{\widehat{\oZ}(\gamma)} \abs{\widehat{\mu_A}(\gamma)}^2 \abs{\widehat{\mu_X}(\gamma)}^{\ell}
~\leq~ \frac{1}{10}.
\]
\fi
Hence, $\ip{g * \mu_V, \mu_A} \geq 4/5$. Expanding this, we get
\ifconf
\begin{align*}
\Ex{a \in A \atop v \in V} {\Ex{a' \in A \atop x_1,\ldots,x_{\ell} \in
  X}{\oZ(a+a'+v+x_1+\ldots+x_{\ell})} }\geq \frac45 \mper
\end{align*}
\else
\[
\E_{a \in A} \E_{v \in V} \E_{a' \in A} \Ex{x_1,\ldots,x_{\ell} \in
  X}{\oZ(a+a'+v+x_1+\ldots+x_{\ell})} \geq 4/5 \mper
\]
\fi
Thus, there exists an $x_0 \in \F_2^n$ such that $\Ex{v \in V}{\oZ(v+x_0)} = \Ex{v \in
  V}{\indicator{Z_{\eta}+x_0}(v)} \geq 4/5$. In other words, there exists $x_0 \in \F_2^n$ such that
$|V \cap (x_0 + Z_{\eta})| > 4|V|/5>|V|/2$. Since any $v\in V$ has $|V|$ representations as $v=
v_1+v_2$ with $v_1,v_2 \in V$, we find that for any $v \in V$, there exist $v_1,v_2 \in x_0 +
Z_{\eta}$ such that $v=v_1+v_2$, and hence $v \in 2 Z_{\eta}$.

The running time is dominated by the $O(n^2\log n \cdot \poly((1/\alpha)^t, \log(1/\gamma')))$ calls made
by the Goldreich-Levin algorithm to the procedure $\Xtest$. For $t = O(\ell^2 \log(\ell/\alpha))$
and $\ell = \log(10/\alpha)$, this implies a running time of $2^{O(\log^4(1/\alpha))} \cdot
\polylog(1/\gamma') \cdot n^3 \log n$.
\end{proofof}


\confomit{
\section{Sumsets of dense sets contain large subspaces}\label{sec:app2}

Inspired by the question of whether dense subsets of $\{1, \dots, N\}$ contain long arithmetic progressions, which has received extensive coverage in the literature \cite{Bourgain:apsiss,Green:apsiss,Sanders:apsiss}, Ben Green asked an analogous question in the finite field setting and obtained the following result \cite{Green:subspace}.

\begin{theorem}[Green's theorem on subspaces in sumsets]\label{thm:green}
Let $A\subseteq \F_2^n$ be a subset of density $\alpha$. Then $A+A$ contains a subspace $V$ of $\F_2^n$ of dimension
\[ \dim(V) =\Omega(\alpha^2 n).\]
\end{theorem}

In \cite{Sanders:half} Sanders showed, using a Fourier-based density-increment strategy, that one can in fact take the subspace $V$ to have dimension $\dim(V) =\Omega(\alpha n)$. He remarks that a bound of similar strength could be obtained via a finite field analogue of the methods of Croot, {\L}aba and Sisask \cite{Croot:2011we}.
Our main theorem in this section is the following, replicating the result from \cite{Croot:2011we}, which falls slightly short of Sanders's bound \cite{Sanders:half}.

\begin{theorem}[Sumsets of dense sets contain large subspaces]\label{thm:sumsets}
Let $A \subseteq \F_2^n$ be a subset of density $\alpha$.  Then $A+A$ contains an affine subspace $V$ of $\F_2^n$ of dimension
\[\dim(V) = \Omega \bigg(\frac{\alpha} {\log^3(1/\alpha)} n\bigg).\]
\end{theorem}

For the proof of the above theorem we shall need a refined version of the almost-periodicity results from Section \ref{sec:ap}. In particular, we shall need the following refined version of Corollary \ref{lem:ap-sumsets-subspace}.

\begin{corollary}[Refined almost-periodicity of sumsets over a subspace]\label{lem:ap-sumsets-subspace-refined}
Let $A\subset\F_2^n$ be a subset of density $\alpha$. Then for every integer $t$ and set $B\subseteq \F_2^n$, there exists a subspace $V$ of codimension $\codim(V) \leq 32 \log (2/\alpha^t)$ with the following property.

For every $v \in V$, for all subsets $S \subseteq \F_2^n$ and for every $\eta, \epsilon >0$ and integer $\ell$,
  \begin{equation}
    \label{eq:subspace-main-refined}
    \Pr_{y\in S}\left[\rho_{A\to B}(y) - \rho_{A\to B}(y+v) \leq \epsilon' \right]\geq 1-16 \frac{\ell} {\eta} \frac{|A+B|}{|S|}\cdot \exp\left(-\eps^2 t/4\right),
  \end{equation}
  where $\epsilon' = 4\eps \ell \sqrt{\rho_{A \to B}(y)}+ 2 \eta +2^{-\ell} \sqrt{|B|/|A|}$.
\end{corollary}

The main difference between the above corollary and Corollary \ref{lem:ap-sumsets-subspace} lies in the term $ \sqrt{\rho_{A \to B}(y)}$ which appears in the expression for $\epsilon'$ in the above corollary. This term makes $\epsilon'$ smaller which in turn makes the above corollary stronger. For the sake of simplicity,
 we only consider in the above corollary one-sided bounds of the form
$\rho_{A\to B}(y) - \rho_{A\to B}(y+v) \leq \epsilon'$ instead of two-sided bounds of the form $\rho_{A\to B}(y) \approx_{\epsilon'} \rho_{A\to B}(y+v)$. This will suffice for the proof of Theorem \ref{thm:sumsets}.

The proof of Corollary \ref{lem:ap-sumsets-subspace-refined} is similar to the proof of Corollary \ref{lem:ap-sumsets-subspace}, and the main difference is that in the proof of Corollary \ref{lem:ap-sumsets-subspace-refined} we perform a more detailed analysis of the distribution $\indicator{B}(a+y)$ when $a$ is distributed uniformly over $A$ and $y$ is a fixed point in $\F_2^n$ and use information on the {\em variance} of this distribution. More specifically, in the proof Corollary \ref{lem:ap-sumsets-subspace-refined}, instead of using the regular Hoeffding bound for sampling (Lemma \ref{lem:hoeffding-sample}), we use the following well-known refinement involving the variance \cite{TaoVu}.

 \begin{lemma}[Refined Hoeffding bound for sampling]\label{lem:hoeffding-variance}
If $\bf X$ is a random variable satisfying $|\bf X-\ex{{\bf X}}| \leq 1$
   and $\hat{\mu}$ is the empirical
average obtained from $t$ samples, then
\[ \prob{\abs{\ex{{\bf X}} - \hat{\mu}} ~>~ \gamma} ~\leq~2 \exp\bigg(- \frac{\gamma^2 t} {4 \sigma^2(X)}\bigg) \]
provided that $\gamma < 2 \sigma^2$.
 \end{lemma} 
 
Since the proof of Corollary \ref{lem:ap-sumsets-subspace-refined} presents some technical complications, we include it in full in Appendix \ref{sec:refap}. The rest of this section is devoted to the proof of Theorem \ref{thm:sumsets} assuming that Corollary \ref{lem:ap-sumsets-subspace-refined} is true.

The idea of the proof of Theorem \ref{thm:sumsets} is as follows. Applying Corollary \ref{lem:ap-sumsets-subspace-refined} with $B= A$ implies the existence of a relatively large subspace $V$ such that for every $v \in V$, for a large fraction of $y \in \F_2^n$,  it holds that $\rho_{A \to A}(y +v) >0$. Our goal will be to show that an affine shift of $V$ is contained in $2A$, or equivalently to show the existence of  an affine shift $y \in \F_2^n$ such that $\rho_{A \to A}(y+v) >0$ for all $v \in V$. Suppose that we have chosen the parameters in  Corollary \ref{lem:ap-sumsets-subspace-refined} in such a way that for every $v \in V$, at least
($1-\delta$)-fraction of $y \in \F_2^n$ satisfy that $\rho_{A \to A}(y+v) >0$. Then the union bound implies that at least a
$(1 - |V| \delta$)-fraction of $y \in \F_2^n$ satisfy the condition $\rho_{A \to A}(y+v) >0$ for all $v \in V$. Thus in order to guarantee the existence of the desired affine shift $y$, it suffices to choose the parameters in Corollary \ref{lem:ap-sumsets-subspace-refined} in such a way that $|V| \delta <1$.

Note that we wouldn't have gained anything from considering the variance in the proof of the quasipolynomial Bogolyubov-Ruzsa lemma (Theorem \ref{thm:strong-bogolyubov}) since there Corollary \ref{lem:ap-sumsets-subspace} is applied to elements $y$ for which $\rho_{A \to B}(y)$ is very large (between 0.9 and 1), and we have no better handle on the variance. In contrast, here the typical element to which we apply Corollary \ref{lem:ap-sumsets-subspace-refined} satisfies $\rho_{A \to A}(y) = \alpha$, so that the variance is small as well.

For the proof of Theorem \ref{thm:sumsets} we shall need the following simple lemma.

\begin{lemma}\label{lem:parabula}
Let $f(t) = t^2 - bt - c$ for $b >0$, $c \geq 0$, and suppose that $0 \leq  t' \leq  t''$ are such that $f(t') > 0$. Then $f(t') \leq f(t'')$.
\end{lemma}

\begin{proof}The fact that $b > 0$, $c \geq 0$ implies that $f(t)$ has a root $t_1 \leq 0$ and another root $t_2 >0$. Thus we have that $f(t)$ is negative in the interval $(0,t_2)$ and is positive in the interval $(t_2, \infty)$.
The fact that $t'  \geq 0$ and $f(t') > 0$ thus implies that $t' > t_2$. The lemma follows by noting that $f$ is monotonically increasing in the interval $(t_2, \infty)$.
\end{proof}

\begin{proofof}{of Theorem \ref{thm:sumsets}}
Apply Corollary \ref{lem:ap-sumsets-subspace-refined} with $B = A$, $\eta = \alpha/24$, $\ell = \log(12/\alpha)$,
$\epsilon = \sqrt{2 \alpha}/(48 \ell)$, $t$ to be determined later on and
\[S = \{ y \in \F_2^n \mid \rho_{A \to A}(y) \geq \alpha/2\}.\]
Noting that
\[\mathbb{E}_{y \in \F_2^n} [ \rho_{A \to A}(y)] = \Pr_{y \in \F_2^n, a \in A}[a+y \in A] =
 \mathbb{E}_{a \in A}\left[\Pr_{y \in  \F_2^n}[a+y \in A]\right]  = \alpha,\]
Markov's inequality implies that $|S| \geq (\alpha/2) \cdot 2^n$.

With this choice of parameters Corollary \ref{lem:ap-sumsets-subspace-refined} implies that for every $v \in V$,
\[ \Pr_{y \in S} \bigg[\rho_{A \to A}(y+v) \geq \rho_{A \to A}(y)  - \alpha/6 - \frac {\sqrt{2 \alpha \cdot \rho_{A \to A}(y)}}
{12} \bigg] \geq 1 - 16 \frac {\ell} {\eta} \cdot \frac {|2A|} {|S|} \cdot \exp( - \epsilon^2 t/4) \]

Let $\delta : = 16 (\ell/\eta) \cdot (|2A|/|S|) \cdot \exp( - \epsilon^2 t/4)$.
Since $\rho_{A \to A}(y) \geq \alpha/2$ for every $y \in S$, the inequality above implies that
\begin{eqnarray*}
 \Pr_{y \in S} \left[\rho_{A \to A}(y+v) \geq \alpha/4 \right]   & = & \Pr_{y \in S} \left[\rho_{A \to A}(y+v) \geq \alpha/2  - \alpha/6 - \frac {\sqrt{2 \alpha \cdot \alpha/2}}
{12} \right] \\
& \geq &  \Pr_{y \in S} \bigg[\rho_{A \to A}(y+v) \geq \rho_{A \to A}(y)  - \alpha/6 - \frac {\sqrt{2 \alpha \cdot \rho_{A \to A}(y)}}
{12} \bigg] \\ & \geq & 1 - \delta
\end{eqnarray*}
where the first inequality follows by applying Lemma \ref{lem:parabula} with $f(t) = t^2 - \big(\sqrt{2\alpha}/12\big)t - \alpha/6 $,
$t' = \sqrt{\alpha/2}$, $t'' = \sqrt{\rho_{A \to A}(y)}$,  and noting that our assumptions imply that $0 \leq t' \leq t''$ and that $f(t') = \alpha/4 >0$.

A union bound then implies that
\[
 \Pr_{y \in S} \left[\rho_{A \to A}(y+v) \geq \alpha/4 \; \forall v \in V \right]  \geq 1 - |V| \cdot \delta.\]
To conclude the proof we shall show that for sufficiently small integer $t$ one can guarantee that $|V| \delta <1$. This in turn will imply the existence of an affine shift $y \in S$ such that $\rho_{A \to A}(y+v) >0$ for every $v \in V$, and consequently $y+V \subseteq  2A$. Our choice of parameters implies that

\begin{eqnarray*}
|V| \delta & = & \bigg( \frac{\alpha^t} {2} \bigg)^{32} \cdot 2^n \cdot  16 \cdot \frac {\ell} {\eta} \cdot \frac {|2A|} {|S|} \cdot \exp( - \epsilon^2 t/4) \\
& \leq & \exp\bigg( -t\bigg(32\log(1/\alpha) +   \frac {2\alpha} {4 \cdot 48^2 \cdot \log^2(12/\alpha)} \bigg)  + \bigg(2 \log(1/\alpha) + n +  \log \log(12/\alpha)  \bigg)\bigg)
\end{eqnarray*}
Thus, $|V| \delta <1$ is guaranteed by letting
\[ t =  \frac{2 \log(1/\alpha) + n + \log \log(12/\alpha)} {32\log(1/\alpha) +   \frac {2\alpha} {4 \cdot 48^2 \cdot \log^2(12/\alpha)}}
= \frac{ n + O (\log(1/\alpha))} {32\log(1/\alpha) +   \Omega(\alpha/\log^2(1/\alpha))}.\]
But for such a choice of $t$ we have that
\[\dim(V)  =  n - 32 \log(1/\alpha) t - 32 = \Omega \bigg(\frac{\alpha} {\log^3(1/\alpha)} n\bigg).\]
\end{proofof}

}

\bibliographystyle{amsalpha}
\bibliography{macros,quadratic}

\confomit{

\newpage

\appendix
\section[Appendix: Proof of Corollary \ref{lem:ap-sumsets-subspace-refined}]{Appendix: Proof of Corollary \ref{lem:ap-sumsets-subspace-refined}}\label{sec:refap}

In order to prove Corollary \ref{lem:ap-sumsets-subspace-refined}, we start with refined versions of Proposition
 \ref{prop:ap-sumsets} and Corollary \ref{lem:ap-sumsets-iterated}, given as Proposition \ref{lem:ap-sumsets-refined} and Corollary
 \ref{lem:ap-sumsets-iterated-refined} below.

\begin{proposition}[Refined version of almost-periodicity of sumsets]\label{lem:ap-sumsets-refined}
Let $A\subset\F_2^n$ be a subset satisfying $|2A|\leq K|A|$. Then for every integer $t$ and set $B\subseteq \F_2^n$ there exists a set $X$ with the following properties.
  \begin{enumerate}
  \item\label{p1} The set $X$ is contained in an affine shift of $A$.
  \item\label{p2} The size of $X$ is at least $|A|/(2K^{t-1})$.

  \item\label{p3} For all $x\in X$ and for all subsets $S \subseteq \F_2^n$,
  \begin{equation}
    \label{eq:refined1}
    \Pr_{y\in S}\left[\rho_{A\to B}(y) -  \rho_{A\to B}(y+x) \leq 2\epsilon \sqrt{\rho_{A\to B}(y)} \right]\geq 1- 8 \frac{|A+B|}{|S|}\cdot \exp\left(- \eps^2 t/4\right).
  \end{equation}
  \end{enumerate}
\end{proposition}

\begin{proof}The proof is very similar to the proof of Proposition \ref{prop:ap-sumsets} and we only point out the differences here.
As in Proposition \ref{prop:ap-sumsets}, let $\rr(y):=\rho_{A\to B}(y)$ and for a vector $\a = (a_1,\ldots, a_t) \in A^t$ let
\[\rra(y)=\frac{\left|\condset{y+a_i\in B}{i=1,\ldots,t}\right|}{t}.\]

For the purpose of this proof, we say that $\a$ is an {\em $\eps$-good estimator} for $y$ if $\rr(y)\approx_{\epsilon'} \rra(y)$ for $\epsilon' = \eps\sqrt{\rho(y)}$ (this is the main point in which this proof differs from the proof of Proposition \ref{prop:ap-sumsets}). Fix $y \in \F_2^n$, and let $Y_i$ be the indicator random variable for the event ``$y+a_i\in B$'' where $a_i$ is chosen uniformly at random from $A$. Then $\rra(y)=\frac{1}{t}\sum_{i=1}^t Y_i$ is the average of $t$ i.i.d. indicator random variables each having mean $\rr(y)$ and variance $\rr(y)(1-\rr(y)) \leq \rr(y)$, so the refined Chernoff-Hoeffding bound (Lemma~\ref{lem:hoeffding-variance}) implies that for all $y \in \F_2^n$,

\begin{equation}
  \Pr_{\a\in A^t}\left[\rr(y)\not\approx_{\eps \sqrt{\rr(y)}} \rra(y)\right]\leq 2 \exp\left(-\eps^2 t/4\right).
\end{equation}

Set $\delta:=2 \exp\left(-\eps^2 t/4\right)$. Similarly to the proof of Proposition \ref{prop:ap-sumsets}, by an averaging argument we get that at least half of the sequences $\a \in A^t$ are $\epsilon$-good estimators for all but ($2\delta$)-fraction of $y \in A+B$, in which case we say that $\a$ is $(\epsilon, 2\delta)$-good estimator for $\rho$.
From here we continue as in the proof of Proposition \ref{prop:ap-sumsets}, letting $\G$ be the set of ($\eps, 2\delta$)-good estimators, that is
\[\G=\condset{\a\in A^t}{\Pr_{y\in A+B}\left[\rr(y)\approx_{\epsilon \sqrt{\rho(y)}} \rra(y)\right]\geq 1-2\delta},\]
and defining $G[\epsilon, 2\delta]_b$, $\hat \a$ and $X$ accordingly.

It can be easily verified that the first two properties listed in the statement are satisfied. Next we show that the third one is satisfied as well.

Suppose $x=\hat{a}_1+a_1$, where $a_1$ is the first element of an $(\eps,2\delta)$-good estimator $\a=(a_1,\ldots,a_t)\in \Gb$.
 Recalling that $\ha$ is an $(\eps,2\delta)$-good estimator,
  we know that for all but a $\left(2\delta |A+B|/|S|\right)$-fraction of $y\in S$,
  \begin{equation}\label{eq:muone-refined}
  \rr(y)\approx_{\eps\sqrt{\rho(y)}} \rraa{\ha}(y).
  \end{equation}
    Similarly, we have that
  \begin{equation}\label{eq:mutwo-refined}
  \rr(y+x)\approx_{\eps\sqrt{\rho(y+x)}} \rraa{\a}(y+x)
  \end{equation}
   for all but a $\left(2\delta |A+B|/|x+S|\right)$-fraction of $y\in S$.
  Using a union bound and the fact that $|S+x| = |S|$, for all but a $\left(4\delta |A+B|/|S|\right)$-fraction of $y\in S$ both (\ref{eq:muone-refined}) and (\ref{eq:mutwo-refined}) hold. For such $y$ we conclude $\rr(y)\approx_{\epsilon'} \rr(y+x)$ for $\epsilon' = \eps\sqrt{\rho(y)} +\epsilon\sqrt{\rho(y+x)}$ using the triangle inequality and the fact that $\rraa{\ha}(y)=\rraa{\ha+x}(y+x)=\rraa{\a}(y+x)$. The proof is completed by noting that $\rho(y)-\rho(y+x) \leq
  2\epsilon \sqrt{\rho(y)}$
holds trivially if $\rho(y+x) \geq \rho(y)$, and hence without loss of generality we may assume that $\rho(y+x) \leq \rho(y)$. This implies in turn that $\epsilon' \leq 2\epsilon \sqrt{\rho(y)}$.

\end{proof}

As before, by an inductive application of Proposition \ref{lem:ap-sumsets-refined} one can prove the following iterated version.

\begin{corollary}[Refined almost-periodicity of sumsets, iterated]\label{lem:ap-sumsets-iterated-refined}
  If $A\subset\F_2^n$ satisfies $|2A|\leq K|A|$ then for every integer $t$ and set $B\subseteq \F_2^n$ there exists a set $X$ with the following properties.
  \begin{enumerate}
  \item\label{p1} The set $X$ is contained in an affine shift of $A$.
  \item\label{p2} The size of $X$ is at least $|A|/(2K^{t-1})$.

  \item\label{p3} For all $x_1, \ldots, x_{\ell} \in X$  and for all subsets $S \subseteq \F_2^n$,
  \begin{equation}
    \label{eq11}
    \Pr_{y\in S}\left[\rho_{A\to B}(y)-  \rho_{A\to B}(y+x_1+\ldots +x_{\ell}) \leq 2\epsilon \ell \sqrt{\rho_{A\to B}(y)}\right]\geq 1-8 \ell \frac{|A+B|}{|S|}\cdot \exp\left(-\eps^2 t/4\right).
  \end{equation}
  \end{enumerate}
\end{corollary}

\begin{proofof}{of Corollary \ref{lem:ap-sumsets-iterated-refined}}
Proposition \ref{lem:ap-sumsets-refined} establishes the case $\ell = 1$. For the induction step, suppose that the lemma holds for some integer $\ell \geq 1$ with a set $X$, and we shall prove that the lemma holds for $\ell+1$ with the same set $X$.

Let $\delta : = 8(|A+B|/|S|)\cdot \exp\left(-\eps^2 t/4\right)$, and fix $x_1, \dots, x_{\ell+1} \in X$.
By the induction hypothesis, for at least ($1- \ell \delta$)-fraction of $y \in S$ it holds that
\begin{equation}\label{eq:refined2}
\rho_{A\to B}(y) -  \rho_{A\to B}(y+x_1+\ldots +x_{\ell}) \leq 2\epsilon\ell\sqrt{\rho_{A\to B}(y)}.
\end{equation}
 The $\ell =1$ case implies that for at least  ($1- \delta$)-fraction of $y \in S$ it holds that
 \begin{equation}\label{eq:refined3}
 \rho_{A\to B}(y+x_1+\ldots+x_{\ell}) -  \rho_{A\to B}(y+x_1+\ldots +x_{\ell+1}) \leq 2\epsilon \sqrt{\rho_{A\to B}(y+x_1+\ldots+x_{\ell}) }.
 \end{equation}

 Thus by union bound we have that at least
($1-(\ell+1)\delta$)-fraction of $y \in S$ satisfy both (\ref{eq:refined2}) and (\ref{eq:refined3}). This implies in turn that for at least
($1-(\ell+1)\delta$)-fraction of $y \in S$ it holds that
\begin{equation}\label{eq:refined4}
 \rho_{A\to B}(y) -  \rho_{A\to B}(y+x_1+\ldots +x_{\ell+1}) \leq 2\epsilon\ell\sqrt{\rho_{A\to B}(y)} + 2\epsilon \sqrt{\rho_{A\to B}(y+x_1+\ldots+x_{\ell}) }.
\end{equation}

If $\rho_{A\to B}(y+x_1+\ldots+x_{\ell}) \leq \rho_{A \to B}(y)$, Equation (\ref{eq:refined4}) implies that
\[ \rho_{A\to B}(y) -  \rho_{A\to B}(y+x_1+\ldots +x_{\ell+1}) \leq 2\epsilon(\ell+1)\sqrt{\rho_{A\to B}(y)}\]
and hence we are done.

Otherwise assume that $\rho_{A\to B}(y+x_1+\ldots+x_{\ell}) \geq \rho_{A \to B}(y)$. Without loss of generality we may also assume that $\rho_{A \to B}(y) - 2\epsilon \sqrt{\rho_{A \to B}(y)} > 0$ since otherwise the fact that $\rho_{A\to B}(y+x_1+\ldots +x_{\ell+1})  \geq 0$ implies that
 \begin{eqnarray*}
\rho_{A\to B}(y+x_1+\ldots +x_{\ell+1})
& \geq & \rho_{A \to B}(y) - 2\epsilon \sqrt{\rho_{A \to B}(y)} \\
& \geq & \rho_{A \to B}(y) - 2\epsilon(\ell+1)\sqrt{\rho_{A \to B}(y)}
\end{eqnarray*}
and hence we are done.

Equation (\ref{eq:refined3}) then implies that
\begin{eqnarray*}
\rho_{A\to B}(y+x_1+\ldots +x_{\ell+1}) & \geq  &\rho_{A\to B}(y+x_1+\ldots +x_{\ell})  - 2\epsilon \sqrt{\rho_{A\to B}(y+x_1+\ldots+x_{\ell}) } \\
& \geq & \rho_{A \to B}(y) - 2\epsilon \sqrt{\rho_{A \to B}(y)} \\
& \geq & \rho_{A \to B}(y) - 2\epsilon(\ell+1)\sqrt{\rho_{A \to B}(y)},
\end{eqnarray*}
where the second inequality follows from Lemma \ref{lem:parabula} by letting $f(t) = t^2 - 2\epsilon t$, $t' = \sqrt{\rho_{A \to B}(y)}$, $t'' = \sqrt{\rho_{A\to B}(y+x_1+\ldots+x_{\ell})}$ and noting that our assumptions imply that $0 \leq t' \leq t''$ and $f(t') >0$.
\end{proofof}

One final ingredient needed for the proof of Corollary \ref{lem:ap-sumsets-subspace-refined} is the following refined version of Lemma \ref{lem:subspace-averaging}.

\begin{lemma}\label{lem:subspace-averaging-refined}
Let $\delta >0$, and let $\eps: (\F_2^n)^{\ell+1} \ra [0,1]$ be an arbitrary function in $\ell+1$ variables. Let $A,B,X,S \subseteq \F_2^n $ be such that for all $x_1,\ldots,x_{\ell} \in X$,
\[ \Pr_{y\in S}\left[\rho_{A\to B}(y) - \rho_{A\to B}(y+x_1+\ldots+x_{\ell}) \leq \epsilon(y,x_1,\ldots,x_{\ell}) \right]\geq 1- \delta.\]
Then for every $\eta >0$ we have
 \[ \Pr_{y\in S}\left[\rho_{A\to B}(y) - \mathbb{E}_{x_1,\ldots,x_{\ell} \in X}[\rho_{A\to B}(y+x_1+\ldots+x_{\ell})] \leq \mathbb{E}_{x_1,\ldots,x_{\ell} \in X}[\epsilon(y,x_1,\ldots,x_{\ell})] + \eta \right] \geq 1- \delta/\eta. \]

  Similarly, if for all $x_1,\ldots,x_{\ell} \in X$,
\[\Pr_{y\in S}\left[ \rho_{A\to B}(y+x_1+\ldots+x_{\ell}) - \rho_{A\to B}(y)  \leq \epsilon(y,x_1,\ldots,x_{\ell}) \right]\geq 1- \delta,\]
  then for every $\eta >0$ we have
 \[ \Pr_{y\in S}\left[ \mathbb{E}_{x_1,\ldots,x_{\ell} \in X}[\rho_{A\to B}(y+x_1+\ldots+x_{\ell})] - \rho_{A\to B}(y)  \leq \mathbb{E}_{x_1,\ldots,x_{\ell} \in X}[\epsilon(y,x_1,\ldots,x_{\ell})] + \eta \right] \geq 1- \delta/\eta.\]
\end{lemma}

\begin{proof}
We shall prove only the first part of the lemma, the second part being almost identical. It follows from Markov's inequality that for at least a $(1-\delta/\eta)$-fraction of $y \in S$, we have
\[\rho_{A\to B}(y) - \rho_{A\to B}(y+x_1+\ldots+x_{\ell}) \leq \epsilon (y,x_1,\ldots,x_{\ell})\]
 for at least a $(1-\eta)$-fraction of $\ell$-tuples $(x_1,\ldots,x_{\ell}) \in X^{\ell}$. Taking expectations, we find that for at least a $(1-\delta/\eta)$-fraction of $y \in S$,
\[\rho_{A\to B}(y)  - \mathbb{E}_{x_1,\ldots,x_{\ell} \in X} \left[\rho_{A\to B}(y+x_1+\ldots+x_{\ell}) \right] \leq \mathbb{E}_{x_1,\ldots,x_{\ell} \in X} [\epsilon(y+x_1+\ldots+x_{\ell}) ] + \eta .\]
\end{proof}

We are now ready for the proof of Corollary \ref{lem:ap-sumsets-subspace-refined}.

\begin{proofof}{of Corollary \ref{lem:ap-sumsets-subspace-refined}}
Again, let $V = \spec_{1/2}(X)^{\perp}$. As before, Chang's theorem (Theorem \ref{thm:chang}) implies that $\codim(V) \leq 32 \log (2/\alpha^t)$.

Let $\delta := 8 \ell (|A+B|/|S|)\cdot \exp\left(-\eps^2 t/4\right)$.
From Lemma \ref{lem:ap-sumsets-iterated-refined} and  Lemma \ref{lem:subspace-averaging-refined} we have that
\begin{equation}\label{eq:refined5}
   \rho_{A\to B}(y) -  \mathbb{E}_{x_1,\ldots,x_{\ell} \in X}[\rho_{A\to B}(y+x_1+\ldots+x_{\ell})] \leq 2\epsilon \ell \sqrt{\rho_{A \to B}(y)} + \eta
\end{equation}
for at least $(1-\delta/\eta)$-fraction of $y \in S$, and similarly that for all $v \in V$,
\begin{eqnarray}\label{eq:refined6}
&&  \mathbb{E}_{x_1,\ldots,x_{\ell} \in X}[\rho_{A\to B}(y+v+x_1+\ldots+x_{\ell})] -  \rho_{A\to B}(y+v) \notag\\
&\leq& 2\epsilon \ell \cdot \mathbb{E}_{x_1,\ldots,x_{\ell} \in X}[\sqrt{\rho_{A \to B}(y+v+x_1+\ldots+x_{\ell})}] + \eta  \nonumber \\
 &\leq& 2\epsilon \ell  \sqrt{ \mathbb{E}_{x_1,\ldots,x_{\ell} \in X}[\rho_{A \to B}(y+v+x_1+\ldots+x_{\ell})] } + \eta 
 \end{eqnarray}
 for at least $(1-\delta/\eta)$-fraction of $y \in S$, where the last inequality is due to convexity.

From Lemma \ref{lem:subspace-fourier} we have that for every $y \in S$ and $v \in V$ it holds that
 \begin{equation}\label{eq:refined7}
\mathbb{E}_{x_1,\ldots,x_{\ell} \in X}[\rho_{A\to B}(y+x_1+\ldots+x_{\ell})] - \mathbb{E}_{x_1,\ldots,x_{\ell} \in X}[\rho_{A\to B}(y+v+x_1+\ldots+x_{\ell})] \leq 2^{-\ell} \sqrt{|B|/|A|}.
\end{equation}

If $ \mathbb{E}_{x_1,\ldots,x_{\ell} \in X}[\rho_{A \to B}(y+v+x_1+\ldots+x_{\ell})] \leq \rho_{A \to B}(y)$ then applying the  union bound to
  (\ref{eq:refined5}), (\ref{eq:refined6}) and (\ref{eq:refined7}) we conclude that
\begin{eqnarray*}
&&  \rho_{A\to B}(y) -  \rho_{A\to B}(y+v)\\
&\leq&  2\epsilon \ell\sqrt{\rho_{A \to B}(y)} +2\epsilon \ell \sqrt{ \mathbb{E}_{x_1,\ldots,x_{\ell} \in X}[\rho_{A \to B}(y+v+x_1+\ldots+x_{\ell})] }+ 2\eta + 2^{-\ell} \sqrt{|B|/|A|} \\
&\leq &  4\epsilon \ell\sqrt{\rho_{A \to B}(y)}  + 2\eta +  2^{-\ell} \sqrt{|B|/|A|}
\end{eqnarray*}
for at least ($1 - 2\delta/\eta$)-fraction of $y \in S$, thus arriving at the desired conclusion.

Otherwise, assume that $ \mathbb{E}_{x_1,\ldots,x_{\ell} \in X}[\rho_{A \to B}(y+v+x_1+\ldots+x_{\ell})] \geq \rho_{A \to B}(y)$. Without loss of generality we may also assume that $\rho_{A \to B}(y) - 2\epsilon \ell \sqrt{\rho_{A \to B}(y)} - \eta >0$ since otherwise we have that
\[\rho_{A \to B}(y+v) \geq \rho_{A \to B}(y) - 2\epsilon\ell  \sqrt{\rho_{A \to B}(y)} - \eta \geq \rho_{A \to B}(y) - \epsilon'\]
and hence we are done. Inequality (\ref{eq:refined6}) then implies that $ \rho_{A \to B} (y+v)$ is at least
\begin{eqnarray*}
 && \mathbb{E}_{x_1,\ldots,x_{\ell} \in X}[\rho_{A \to B}(y+v+x_1+\ldots+x_{\ell})] - 2\epsilon \ell \sqrt{ \mathbb{E}_{x_1,\ldots,x_{\ell} \in X}[\rho_{A \to B}(y+v+x_1+\ldots+x_{\ell})] } - \eta \\
 &\geq&
\rho_{A \to B}(y) -2\epsilon \ell\sqrt{\rho_{A \to B}(y)} - \eta \\
&\geq& \rho_{A \to B}(y) - \epsilon'
\end{eqnarray*}
where the first inequality follows from Lemma \ref{lem:parabula} by letting $f(t) = t^2 - 2\epsilon \ell t - \eta$, $t' = \sqrt{\rho_{A \to B}(y)}$, $t'' = \sqrt{\mathbb{E}_{x_1,\ldots,x_{\ell} \in X}[\rho_{A \to B}(y+v+x_1+\ldots+x_{\ell})]}$ and noting that our assumptions imply that $0 \leq t' \leq t''$ and $f(t') >0$.
\end{proofof}

}

\end{document}